\pdfoutput=1

\documentclass[11pt]{article}

\usepackage[left=1.25in, top=1in, bottom=1in, right=1.25in]{geometry}
\usepackage[colorlinks=True,citecolor=ForestGreen]{hyperref}
\usepackage[usenames,dvipsnames]{xcolor}
\usepackage[switch,mathlines]{lineno}
\usepackage{natbib}
\usepackage{newtxtext}

\usepackage{latexsym}

\usepackage[T1]{fontenc}

\usepackage[utf8]{inputenc}

\usepackage{microtype}

\usepackage{inconsolata}

\usepackage{url}            %
\usepackage{booktabs}       %
\usepackage{amsfonts}       %
\usepackage{nicefrac}       %

\usepackage{subcaption}
\usepackage{algorithm}
\usepackage{algorithmic}

\usepackage[utf8]{inputenc} %
\usepackage[T1]{fontenc}    %
\usepackage{url}            %
\usepackage{booktabs}       %
\usepackage{amsfonts}       %
\usepackage{nicefrac}       %
\usepackage{microtype}      %
\usepackage{dsfont}
\usepackage{soul}
\usepackage{natbib}

\usepackage{float}
\usepackage[font=footnotesize,format=hang]{caption}
\usepackage{graphicx}
\usepackage[export]{adjustbox}
\usepackage[inline]{enumitem}

\usepackage{multirow}

\usepackage{tikz}
\usetikzlibrary{positioning}

\usepackage{amsmath,amssymb,amsthm}
\usepackage[mathscr]{eucal}
\usepackage{stmaryrd}
\usepackage{accents}
\usepackage{upgreek}

\usepackage[most]{tcolorbox}
\usepackage{setspace}

\usepackage{booktabs,colortbl,multirow}

\usepackage{pifont}

\graphicspath{{figs//}}

\allowdisplaybreaks

\newcommand{\argmin}{{\operatorname{argmin }}}

\newcommand{\R}{\mathbf{R}}

\newcommand{\Real}{\mathbb{R}}

\renewcommand{\R}{\mathbb{R}}

\newcommand{\best}[1]{{\cellcolor{gray!25}{#1}}}

\newcommand{\cneeded}[2]{\textcolor{orange}{[citation-needed]}}

\usepackage[scaled=0.8]{FiraMono}
\usepackage[capitalize,noabbrev]{cleveref}

\theoremstyle{plain}
\newtheorem{theorem}{Theorem}[section]
\newtheorem{proposition}[theorem]{Proposition}
\newtheorem{lemma}[theorem]{Lemma}

\theoremstyle{definition}

\theoremstyle{remark}

\usepackage[textsize=tiny]{todonotes}

\newcommand{\mytitle}{Efficient Document Ranking with Learnable Late Interactions}

\title{\mytitle{}}

\author{
Ziwei Ji \and Himanshu Jain \and Andreas Veit \and Sashank J. Reddi \and Sadeep Jayasumana \and Ankit Singh Rawat \and Aditya Krishna Menon \and Felix Yu \and Sanjiv Kumar \\ \\
Google \\
\texttt{\{ziweiji,himj,aveit,sashank,sadeep,ankitsrawat,\phantom{\}}}\\ \texttt{\phantom{\{}adityakmenon,felixyu,sanjivk\}\ @google.com}
}

\begin{document}
\maketitle

\begin{abstract}
    Cross-Encoder (CE) and Dual-Encoder (DE) models are two fundamental approaches for 
predicting
query-document relevance in information retrieval. 
To predict relevance, 
CE models use \emph{joint} query-document embeddings, 
while
DE models maintain \emph{factorized} query-document embeddings;
usually, the former has higher quality while
the latter has lower latency.
Recently, \emph{late-interaction} models have been proposed to realize more favorable latency-quality trade-offs, by using a DE structure followed by a lightweight scorer based on query and document token embeddings. 
However, 
these lightweight scorers are often hand-crafted, and there is no understanding of their approximation power;
further, such scorers require access to individual document token embeddings, which imposes an increased latency and storage burden over DE models.
In this paper, we propose novel \emph{learnable} late-interaction models (LITE) that resolve these issues.
Theoretically,
we prove that LITE is a universal approximator of continuous scoring functions, 
even for relatively small embedding dimension.
Empirically, LITE outperforms previous late-interaction models such as ColBERT on both in-domain and zero-shot re-ranking tasks.
For instance, experiments on MS MARCO passage re-ranking show that
LITE not only yields a model with better generalization, but also lowers latency and requires $0.25\times$ storage compared to ColBERT.

\end{abstract}

\section{Introduction}

\begin{figure*}[!t]
    \begin{subfigure}{0.3\textwidth}
        \vspace{19mm}
        \centering
        \includegraphics[scale=0.35]{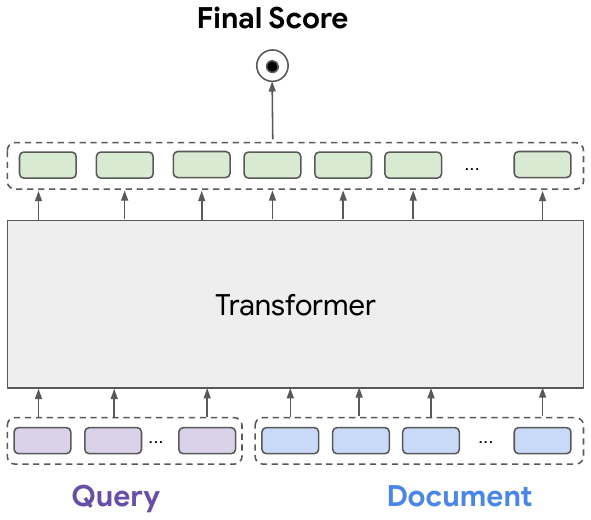}
        \caption{Cross Encoder (CE)}
        \label{fig:ce}
    \end{subfigure}%
    \begin{subfigure}{0.3\textwidth}
        \vspace{12.5mm}
        \centering
        \includegraphics[scale=0.35]{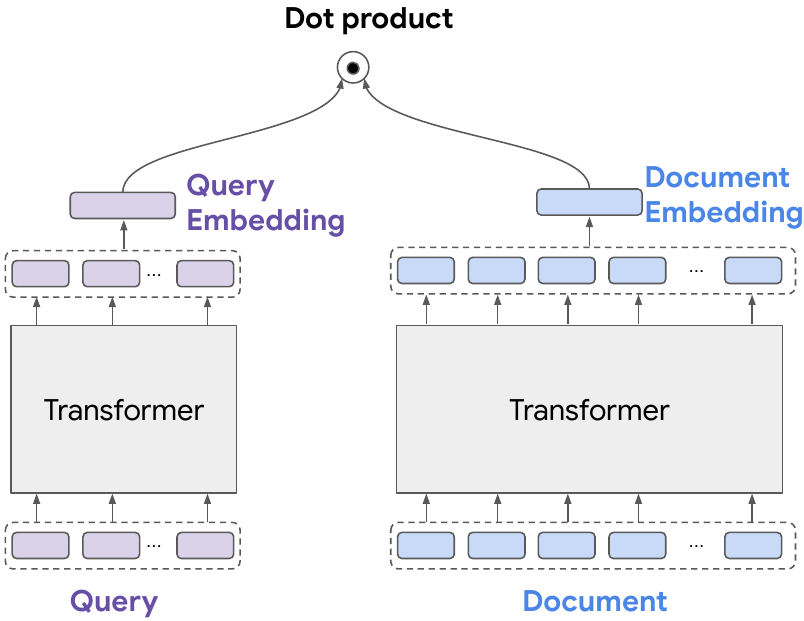}
        \caption{Dual Encoder (DE)}
        \label{fig:de}
    \end{subfigure}%
    \qquad
    \begin{subfigure}{0.3\textwidth}
        \centering
        \includegraphics[scale=0.35]{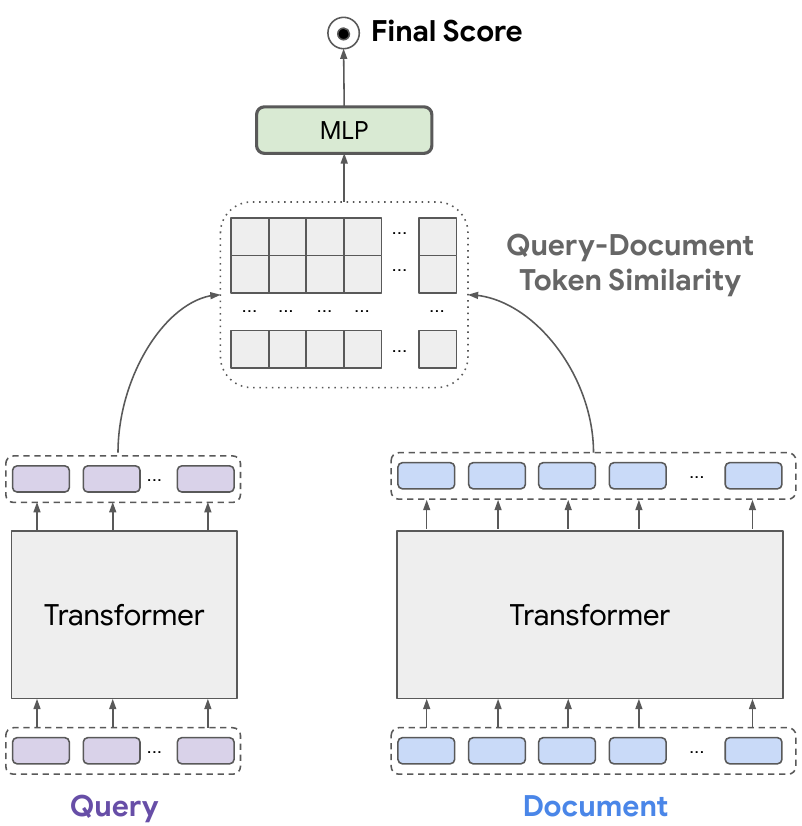}
        \caption{LITE}
        \label{fig:lite}
    \end{subfigure}
    \caption{Illustration of different query-document relevance models. (a) CE  models compute a joint query-document embedding by passing the concatenated query/document tokens through a single Transformer. (b) In DE models, query and document embeddings are computed separately with their respective Transformers and the relevance score is the dot product of these embeddings. (c) In the proposed LITE method, query and document token embeddings are computed similarly to DE, but instead of a dot product, we first compute the similarity matrix between each pair of query and document tokens, and pass this matrix through an MLP to produce the final relevance score.}
    \vspace{-\baselineskip}
\end{figure*}

Transformers~\citep{Vaswani:2017} have emerged as a successful model for information retrieval problems, where the goal is to retrieve and rank relevant documents for a given query~\citep{Nogueira:2019}.
Two families of Transformer-based models are popular: 
\textit{cross-encoder} (CE) and \textit{dual-encoder} (DE) models.
Given a (query, document) pair, 
CE models operate akin to a BERT-style encoder~\citep{Devlin:2019}: the query and document are concatenated, and sent to a Transformer encoder which outputs a relevance score (cf.~Figure~\ref{fig:ce}).
CE models can learn complex query-document relationships, as they allow for cross-interaction between query and document tokens. 

By contrast, DE models apply two separate Transformer encoders to the query and document, respectively, producing separate query and document embedding vectors~\citep{Reimers:2019}.
The dot product of these two vectors is used as the final relevance score~(cf.~Figure~\ref{fig:de}).
Compared to CE models,
DE models are usually less accurate~\citep{Hofstatter:2020},
since the only interaction between the query and document occurs in the final dot product.
However, DE models have much lower latency,
since all the document embedding vectors can be pre-computed offline.

Recently, \textit{late-interaction} models have provided alternatives with a more favorable latency-quality trade-off compared to CE and DE models.
Similarly to DE models, late-interaction models also use a two-Transformer structure, but they store more information and employ additional nonlinear operations to calculate the final score.
In particular, let $\mathbf{Q}\in\R^{P\times L_1}$ and $\mathbf{D}\in\R^{P\times L_2}$ denote the query and document token embeddings output by the two Transformers, i.e., there are $L_1$ query token embedding vectors and $L_2$ document token embedding vectors of dimension $P$.
DE models simply pool $\mathbf{Q}$ and $\mathbf{D}$ into two vectors, and take the dot product.
By contrast, ColBERT \citep{Khattab:2020} calculates the (token-wise) similarity matrix $\mathbf{Q}^\top\mathbf{D}$ and computes the final score via a sum-max reduction $\sum_i\max_j(\mathbf{Q}^\top\mathbf{D})_{i,j}$.

While the sum-max score reduction lets ColBERT achieve better accuracy than DE, it is unclear whether this hand-crafted reduction can capture arbitrary complex query-document interactions.
Moreover, ColBERT can have higher latency than DE: calculating the similarity matrix $\mathbf{Q}^\top\mathbf{D}$ requires $L_1 \cdot L_2$ dot products, while the DE model only requires one dot product.
Additionally, to reduce online latency, ColBERT needs to pre-compute and store the Transformer embedding matrix $\mathbf{D}$ for each document~\cite{Hofstatter:2020,Santhanam:2022}.
This can entail significant storage space if we decide to store a large number of document tokens, since there can be billions of documents in industry-scale information retrieval systems~\citep{zhang2013image,overwijk2022clueweb22}.
(See \Cref{sec:limitations} for a detailed discussion.)

To reduce latency and storage cost, one may seek to store fewer document tokens, and/or reduce the dimension of each token embedding vector.
However, it is unclear how these influence performance.
In fact, such reduction can significantly hurt the accuracy of ColBERT, as we show in \Cref{sec:limited_storage}.

\textbf{Contributions.}
In this work, we propose \textit{\textbf{li}ghtweight scoring with \textbf{t}oken \textbf{e}insum} (LITE), which addresses the aforementioned shortcomings of existing late-interaction models.
LITE applies a \emph{lightweight and learnable non-linear transformation} on top of Transformer encoders, which corresponds to processing the (token-wise) similarity matrix $\mathbf{S} = \mathbf{Q}^\top\mathbf{D}$ via shallow multi-layer perceptron (MLP) layers (cf.~Figure~\ref{fig:lite} and \Cref{sec:lite}).
In particular, we focus on a \textit{separable LITE} scorer which applies two shared MLPs to the rows and the columns of $\mathbf{S}$ (in that order), and then projects the resulting matrix to a single scalar.

Theoretically, we rigorously establish the expressive power of LITE: we show that LITE is a universal approximator of continuous scoring functions in $\ell_2$ distance, even under tight storage constraints (cf.~\Cref{fact:lite_univ_approx}).
To our knowledge, this is the \emph{first formal result about the approximation power of late-interaction methods}.
Further, we also construct a scoring function that cannot be approximated by a DE model with restricted embedding dimension (cf. \Cref{fact:de_neg}).

Empirically, we show that LITE can systematically improve upon existing late-interaction methods like ColBERT on both in-domain benchmarks such as MS MARCO and Natural Questions (cf. \Cref{tab:msmarco_mrr}), and out-of-domain benchmarks such as BEIR (cf. \Cref{tab:beir_ndcg}).
Moreover, LITE can be much more accurate than ColBERT while having lower latency and storage cost (cf. \Cref{tab:latency}).

\section{Background}

Given a query $q \in \mathscr{Q}$, the goal of information retrieval~\citep{Mitra:2018} is to identify the set of \textit{relevant documents} from some corpus $\mathscr{D}$.
Typically, $| \mathscr{D} |$ is large (e.g., $\mathcal{O}(10^9)$), while the number of relevant documents is small (e.g., $\mathcal{O}(10)$).
A classical strategy employs a two-phase approach: in the \textit{retrieval} phase, for moderate $K$ (e.g., $\mathscr{O}( 10^3 )$), one retrieves the top-$K$ documents based on a scoring function $s_{\rm ret} \colon \mathscr{Q} \times \mathscr{D} \to \Real$. %
These retrieved documents may potentially include some irrelevant documents.
In the \textit{re-ranking} phase, one applies $s_{\rm rr} \colon \mathscr{Q} \times \mathscr{D} \to \Real$ to re-score the $K$ documents, and keep the top scoring ones.

While $s_{\rm ret}$ and $s_{\rm rr}$ both score query-document relevance, they are often implemented via fundamentally different techniques.
Efficiency is more important for $s_{\rm ret}$ since we need to evaluate it over \textit{all} documents; models such as TF-IDF and BM25~\citep{Robertson:2009} and approximate nearest neighbor search \citep{guo2016quantization,johnson2019billion,Guo:2020} are used for this purpose.
On the other hand, in the second phase we usually only need to re-score a few ($K \ll | \mathscr{D} |$) documents, and thus we can usually get higher accuracy by using more expensive models for $s_{\rm rr}$.
In this work, we focus on re-ranking.

\subsection{Cross- and Dual-Encoders}

\textit{Transformers}~\citep{Vaswani:2017} have been explored for both retrieval and re-ranking.
Given a finite set $\mathscr{X}$, a Transformer is a function $T \colon \mathscr{X}^L \to \Real^{P \times L}$, where $L$ is the sequence length and $P$ is the embedding size of each token in the sequence.
A simplified Transformer network is introduced in \Cref{sec:lite_univ_approx} and used in our universal approximation results; for more details, we refer the readers to \citep{Vaswani:2017,Devlin:2019}.

To estimate query-document relevance via Transformers, one first \textit{tokenizes} the query and document (e.g., using a SentencePiece tokeniser~\citep{Kudo:2018}) into $q = (q_1, \ldots, q_{L_1})$ and $d = (d_1, \ldots, d_{L_2})$.
There are then two basic strategies.
In \textit{cross-encoder} (CE) models~\citep{Nogueira:2019}, we apply a single Transformer to the concatenation of $q$ and $d$, and estimate relevance with learned weights $\mathbf{w}$:
\begin{equation}
    \label{eqn:ce}
    s( q, d ) = \mathbf{w}^\top {\sf pool}( T( {\sf concat}(q, d) ) ),
\end{equation}
where ${\sf pool}$ denotes a pooling strategy by which we reduce a sequence of Transformer token embeddings into a single vector.
CE models can often achieve high accuracy since they can take into account interactions between the query and document tokens in \textit{every} Transformer layer.
However, they can also be expensive at inference time: we need to compute \eqref{eqn:ce} for all retrieved documents, each of which involves an expensive Transformer inference (see \Cref{sec:latency} for concrete evaluations).

By contrast, in \textit{dual-encoder} (DE) models~\citep{Karpukhin:2020}, we apply separate Transformers 
$T_1, T_2$ to the query and document, and then compute
\begin{equation}
    \label{eqn:de_dp}
    s( q, d ) = {\sf pool}( T_1( q ) )^\top {\sf pool}( T_2( d ) ).
\end{equation}
In practice, DE is usually less accurate than CE for re-ranking~\citep{Hofstatter:2020}, since the only interaction between the query and document is the final dot product.
Using stronger $T_1$ and $T_2$ can increase the accuracy of DE \citep{t5_t3,ma2023fine}, but it is also more expensive.
On the other hand, since all document embeddings ${\sf pool}(T_2(d))$ can be pre-computed offline, DE has much lower latency than CE with the same embedding backbone.

Another idea is to apply an MLP to the concatenation of ${\sf pool}( T_1( q ) )$ and ${\sf pool}( T_2( d ) )$ \citep{he2017neural}.
However,~\citet{rendle2020neural} claim that it may not be better than dot-product DE, partly because it is non-trivial to learn the dot-product operation with an MLP given the concatenated query-document embedding as the input.

\subsection{Late-interaction scorers}

Recently, there has been interest in \textit{late-interaction} models.
Similarly to DE models, such models also embed queries and documents separately into $T_1(q)$ and $T_2(d)$; however, they do not use pooling operations, but instead calculate dot products between \emph{all pairs} of query and document token embeddings, and perform a non-linear score reduction.
Formally, let us define query and document Transformer embeddings $\mathbf{Q} = ( \mathbf{q}_1, \ldots, \mathbf{q}_{L_1} ) := T_1( q ) \in \Real^{P \times L_1}$ and $\mathbf{D} = ( \mathbf{d}_1, \ldots, \mathbf{d}_{L_2} ) := T_2( d ) \in \Real^{P \times L_2}$, and let $\mathbf{S}:=\mathbf{Q}^\top\mathbf{D}$ denote the similarity matrix.
\textit{ColBERT}~\citep{Khattab:2020} then performs a non-linear sum-max reduction of $\mathbf{S}$:
$$ s( q, d ) = \sum\nolimits_{i \in [L_1]} \max\nolimits_{j \in [L_2]} \mathbf{q}_i^\top \mathbf{d}_j. $$
This non-linearity allows ColBERT to achieve better accuracy than DE.
See~\citep{Yi:2021} for a related model.
Another similar approach is CEDR~\citep{MacAvaney:2019}, which uses multiple query-document similarity matrices (one for each layer) from pre-trained Transformers.
For each query token, instead of only using the most aligned document token, \citet{qian2022multi} suggest considering the top-$k$ aligned document tokens.

Instead of using similarities between all pairs of query and document token embeddings, COIL \citep{gao2021coil} only considers pairs of query and document tokens that have the same token ID, while CITADEL \citep{li2022citadel} further implements a dynamic lexical routing.
\citet{li2023slim} use sparse token representations that can achieve competitive accuracy compared to ColBERT while being much faster.
\citet{mysore2021multi} suggest using co-citations as supervision for training.

Late-interaction models have precedent in the classical IR literature.
For example, DRMM~\citep{Guo:2016} scores (query, document) relevance using a feedforward network on top of \textit{count histogram} features.
On top of the query-document token similarity matrix based on Word2Vec, MatchPyramid~\citep{match_pyramid} applies a convolutional network, while KNRM~\citep{Xiong:2017} performs kernel-based pooling.
ConvKNRM~\citep{Dai:2018ConvKNRM} further uses a convolutional network on top of learned token embeddings to produce contextual embeddings.
There are also relevant models from the collaborative filtering literature, such as~\citet{Dziugaite:2015}.

\subsection{Limitations of existing late-interaction scorers}
\label{sec:limitations}

Late-interaction scorers such as ColBERT may be used in both the retrieval and re-ranking phases.
In this paper, we focus on the latter, which has been considered in several previous works, e.g.,~\citep{Hofstatter:2020,Santhanam:2022,Ren:2021}.
While ColBERT can yield a more favourable latency versus quality trade-off compared to DE and CE models, there are two important limitations for its use in re-ranking.

\emph{Limited expressivity of hand-crafted reductions}.
Although prior late-interaction models include more non-linearity compared with DE, they rely on hand-crafted score reductions, such as sum-max in ColBERT.
It is unclear if these operations can capture arbitrary complex interactions among query and document tokens that define the true relevance. 

\emph{Latency and storage overhead}.
Compared with CE, both DE and late-interaction models reduce latency by relying on pre-computed document (token) embeddings.
For DE, this requires storing a single document embedding vector (after proper pooling, cf. \eqref{eqn:de_dp}), and during online inference, we need to take one dot product.
Unfortunately, for late-interaction models, the latency and storage cost can be much higher: suppose we use $L_1$ query embedding vectors and $L_2$ document embedding vectors to calculate the similarity matrix, then the storage cost is $L_2$ times larger than that of DE models\footnote{The document (token) index can be stored on disk, or in RAM.
Storing in RAM significantly reduces latency, as we do not need to pay the cost of transferring embeddings from disk.
Even if one were to store the index on disk, 
it is still of interest to reduce the total embedding size to reduce the storage and transfer cost/latency (which would scale linearly with embedding size). }, and we need to take $L_1L_2$ dot products to obtain the similarity matrix.
It is unclear how various ways to reduce the latency and storage cost affect the model performance.

In the next section, we present \emph{LITE}, a novel late-interaction scorer that addresses both aforementioned shortcomings: 
(1) LITE can provably approximate a broad class of ground truth scoring functions (cf. \Cref{fact:lite_univ_approx}); and 
(2) it is more accurate than prior late-interaction methods on both in-domain and zero-shot tasks, and is amenable to latency and storage reduction with graceful degradation in model performance (cf. \Cref{sec:experiments}).

\section{LITE scorers}
\label{sec:lite}

We now introduce LITE scorers.
Let $\mathbf{S}:= \mathbf{Q}^\top \mathbf{D} \in \mathbb{R}^{L_1 \times L_2}$ denote the similarity matrix which consists of the dot products of all query-document Transformer token embedding pairs. 
LITE models apply MLPs to reduce $\mathbf{S}$ to a scalar score.
A natural option is to flatten $\mathbf{S}$ and then apply an MLP; we call this \emph{flattened LITE}. 
On the other hand, in this paper we focus on another MLP model which we call separable LITE, motivated by separable convolution \citep{xception} and MLP-Mixer \citep{mlp_mixer}: we first apply row-wise updates to $\mathbf{S}$, then column-wise updates, and then a linear projection to get a scalar score.
Formally, we first calculate $\mathbf{S}',\mathbf{S}^{''} \in \mathbb{R}^{L_1 \times L_2}$ as follows: 
for all $1\le i\le L_1$ and $1\le j\le L_2$, let
\begin{align}
    \label{eq:sep_lite_row}
    \mathbf{S}'_{i,:} &= {\sf LN}(\sigma(\mathbf{W}_2 {\sf LN}(\sigma(\mathbf{W}_1 \mathbf{S}_{i,:} + \mathbf{b}_1))+\mathbf{b}_2)), \\
    \label{eq:sep_lite_col}
    \mathbf{S}''_{:,j} &= {\sf LN}(\sigma(\mathbf{W}_4 {\sf LN}(\sigma(\mathbf{W}_3 \mathbf{S}'_{:,j} + \mathbf{b}_3))+\mathbf{b}_4)),
\end{align}
where ${\sf LN}$, $\sigma$ respectively denote layer-norm and ReLU.
The final score is given by $\mathbf{w}^\top {\sf vec} (\mathbf{S}'')$.

Given the above definitions, it is natural to consider the expressivity of LITE.
In particular, there are two fundamental questions:
(1) Can we always approximate (continuous) scoring functions using LITE, even though LITE only has the similarity matrix as inputs and the original Transformer embeddings are lost?
(2) Are LITE models more expressive than simpler models such as DE?

We answer these questions in the following: we show that LITE models are universal approximators of continuous scoring functions (cf. \Cref{fact:lite_univ_approx}), while there exists a scoring function which cannot be approximated by a simple dot-product DE (cf. \Cref{fact:de_neg}).

\subsection{Universal approximation with LITE}
\label{sec:lite_univ_approx}

We consider the Transformer architecture described by \citep{Yun:2020}: it includes multiple encoding layers, each of them can be parameterized as $ {\sf A}( \mathbf{X} ) + {\sf FF}( {\sf A}( \mathbf{X} ) ), $
where $\mathbf{X}\in\R^{P\times L}$ denotes the input, ${\sf FF}$ denotes a feedforward network, and ${\sf A}(\mathbf{X})$ denotes an \textit{attention} block:
$$\mathbf{X} + \sum_{i=1}^H\mathbf{W}^{i}_{\rm o} \mathbf{W}^{i}_{\rm v} \mathbf{X} {\sf Softmax}( (\mathbf{W}^{i}_{\rm k} \mathbf{X})^\top (\mathbf{W}^{i}_{\rm q} \mathbf{X}) ). $$
Here $\mathbf{W}^i_{\rm q}, \mathbf{W}^i_{\rm k}, \mathbf{W}^i_{\rm v}\in\R^{C\times P}$ are \textit{query}, \textit{key} and \textit{value} and projection matrices, $\mathbf{W}^i_{\rm o}\in\R^{P\times C}$ are output projection matrices, and $H,C$ denotes the number of heads and dimension of each head.
The ${\sf Softmax}$ function is applied to each input column. 

A Transformer network defined in the above way is \textit{permutation-equivariant} \citep[Claim 1]{Yun:2020}: if we permute the input token sequence, then the output token sequence is permuted in the same way.
If we want the network to distinguish between different orders of tokens, we can add a positional encoding matrix $\mathbf{E}\in\R^{P\times L}$ to the input $\mathbf{X}$, and apply a Transformer network to $\mathbf{X}+\mathbf{E}$.
    
As discussed in previous sections, in the late-interaction setting, we may need to store the whole Transformer output with shape $P\times L$, which can be expensive.
One solution is to apply a pooling function to reduce the number of tokens; we empirically study this method in \Cref{sec:limited_storage}, and in \Cref{fact:lite_univ_approx}, we apply pooling functions to map the Transformer output in $\R^{P\times L}$ to $\R^{P\times2}$, i.e., a sequence of two token embeddings.
We show that two query tokens and two document tokens are enough for universal approximation.

Next, we define the scorers.
Let $\mathcal{F}_{\sigma, n}$ denote the set of 2-layer ReLU networks with $n$-dimensional inputs and a scalar output:
\begin{align*}
    \mathcal{F}_{\sigma,n}:=\left\{\mathbf{z}\to\mathbf{a}^\top\sigma(\mathbf{W}\mathbf{z}+\mathbf{b})\right\},
\end{align*}
where $\sigma$ denotes the ReLU activation, $\mathbf{z}\in\R^n$, $\mathbf{W}\in\R^{m\times n}$, $\mathbf{a},\mathbf{b}\in\R^m$, and we allow $m$ to be arbitrarily large.
We first consider a class of flattened LITE scorers, including all two-layer ReLU networks on top of $\mathbf{S}$ that output a scalar score:
\begin{align*}
    \mathcal{F}_{\rm f}:=\left\{\mathbf{S}\to f({\sf vec}(\mathbf{S}))\middle|f\in\mathcal{F}_{\sigma,L_1\cdot L_2}\right\}.
\end{align*}
For separable LITE, we consider a simplified version of \eqref{eq:sep_lite_row} and \eqref{eq:sep_lite_col}, but without loss of generality, as described below: we first use a 2-layer ReLU network $f_1:\R^{L_2}\to\R$ to reduce every row of $\mathbf{S}$ to a single scalar, and thus transform $\mathbf{S}$ into a column vector; and then we apply another 2-layer ReLU network $f_2$ to reduce this column vector into a scalar.
Formally, 
\begin{align*}
    \mathcal{F}_{\rm s}:=\left\{\mathbf{S}\to f_2(f_1(\mathbf{S})) \middle|f_1\in\mathcal{F}_{\sigma,L_2},f_2\in\mathcal{F}_{\sigma,L_1}\right\},
\end{align*}
where we let $f_1(\mathbf{S})\in\R^{L_1}$ denote the result of applying $f_1$ to every row of $\mathbf{S}$.
Note that $\mathcal{F}_{\rm s}$ is a subset of the function class defined by \eqref{eq:sep_lite_row} and \eqref{eq:sep_lite_col} (ignoring layer normalization).

Here is our universal approximation result.
\begin{theorem}[Universal approximation with LITE]\label{fact:lite_univ_approx}
    Let $s:\R^{(P\times L_1)\times(P\times L_2)}\to\R$ denote a continuous scoring function with a compact support $\Omega$ and $L_1,L_2\ge2$.
    For any $\mathcal{F}\in\{\mathcal{F}_{\rm f},\mathcal{F}_{\rm s}\}$ and any $\epsilon>0$, there exist a scorer $f\in\mathcal{F}$, and $T_1:\R^{P\times L_1}\to\R^{P\times2}$ and $T_2:\R^{P\times L_2}\to\R^{P\times2}$, both of which consist of positional encodings, a Transformer and a pooling function, such that 
    \begin{align*}
        \int_\Omega(f(T_1(\mathbf{X})^\top T_2(\mathbf{Y}))-s(\mathbf{X},\mathbf{Y}))^2{\rm d}(\mathbf{X},\mathbf{Y})\le\epsilon.
    \end{align*}
\end{theorem}

The proof 
is given in \Cref{sec:lite_univ_approx_proof}, and is based on the ``contextual mapping'' techniques from \citep{Yun:2020}. 
This result is non-trivial, since the input to LITE scorers is the similarity matrix based on only two query tokens and two document tokens; this means LITE models are universal approximators even under strong constraints on the total embedding size.
In contrast, as we show in \Cref{fact:de_neg}, if the total embedding size is less than $P \cdot L$, then a dot-product DE can have a large approximation error.

\subsection{Non-universality of existing scorers}

In addition to \Cref{fact:lite_univ_approx}, even without positional encodings, in \Cref{fact:lite_univ_approx_full} we show that LITE scorers are still universal approximators of arbitrary continuous scoring functions if we do not apply pooling.
By contrast, without positional encodings, ColBERT can only represent permutation-equivariant ground-truth scoring functions, because the summation and maximum operations do not consider the order of input tokens.
It is an open question if ColBERT is a universal approximator with positional encodings.

If we ask whether a dot-product DE can approximate arbitrary continuous functions, then we give a negative result.

\begin{theorem}[Limitation of DE with restricted embedding dimension]
\label{fact:de_neg}
Suppose each query and document both have $L\ge2$ tokens. 
There exists a continuous ground-truth scoring function $s$ supported on $\Omega:=[0,1]^{P\times L}\times[0,1]^{P\times L}$, such that if $O\le P \cdot L-1$, then for any mappings $h_1,h_2:\R^{P\times L}\to\R^O$ that map queries and documents to $O$-dimensional vectors respectively,
\begin{equation*}
    \int_\Omega(h_1(\mathbf{X})^\top h_2(\mathbf{Y})-s(\mathbf{X},\mathbf{Y}))^2{\rm d}(\mathbf{X},\mathbf{Y})\ge\frac{1}{20}.
\end{equation*}
\end{theorem}

Previously \citet{Menon:2022DE} showed that if there is no constraint on the embedding dimension, then dot-product DE is a universal approximator of continuous functions.
By contrast, here we show if the DE embedding dimension is less than $P \cdot L$, there could be a constant approximation error.

\section{Experiments}
\label{sec:experiments}

We now evaluate the proposed LITE scorer on a few standard information retrieval benchmarks, where we confirm that LITE significantly improves accuracy over existing DE and late-interaction methods on both in-domain and out-of-domain tasks.
Moreover, we show that LITE remains competitive as we reduce the latency and storage cost, and in particular, LITE can achieve higher accuracy than ColBERT with less latency and $0.25\times$ storage cost.

\subsection{Experimental setup}

\paragraph{Datasets.}

We evaluate scorers on both in-domain re-ranking on the MS~MARCO~\citep{Nguyen:2016} and Natural Questions~(NQ;~\citep{Kwiatkowski:2019}) datasets,
and zero-shot re-ranking on the BEIR~\citep{thakur2021beir} dataset.

\paragraph{Training.}

For training on MS MARCO, we use the official training set of triplets $(q,d_+,d_-)$, where document $d_+$ is relevant to query $q$ while $d_-$ is irrelevant. 
State-of-the-art methods on MS MARCO also use hard-negative mining~\citep{Qu:2021,Santhanam:2022}; however, in this paper our focus is on comparing different late-interaction scorers, and thus we simply use the original triplet training data.

We use labels from a CE teacher model during training, as it has been observed that distillation can significantly improve performance~\citep{Santhanam:2022,Menon:2022DE}.
For MS MARCO, we use the scores from the {\tt T2} teacher released by~\citet{Hofstatter:2020}.
For the NQ dataset, we use a teacher model trained with 19 hard-negatives mined with BM25, following~\citep{Menon:2022DE}.
For loss functions, we try the KL loss and the margin MSE loss (see \Cref{sec:loss_metric} for definitions of loss functions and more details of training).

\paragraph{Evaluation.}

For MS MARCO, we use the standard Dev set and the TREC DL 19 and 20 test sets \citep{trec_19,trec_20}.
For NQ, we utilize the version of this dataset used in \citep{Karpukhin:2020}, which consists of questions, positive passages containing the correct answer, and a collection of Wikipedia passages. 
Re-ranking metrics are reported on the Dev query set with 200 passages containing positives, 100 BM25 hard-negatives and up to 100 random negatives, following~\citep{Menon:2022DE}.
We report MRR@10 \citep{Radev:2002} and nDCG@10 \citep{jarvelin2002cumulated} scores.

For BEIR, following \citep{thakur2021beir}, we take the scorers trained on MS MARCO and evaluate zero-shot transfer performance.
Specifically, we report evaluation results on the 14 public datasets.
\citet{thakur2021beir} evaluate the CE model by first retrieving 100 documents using BM25, and then calculating the nDCG@10 score for CE re-ranking; we use the same procedure.

\paragraph{Models.}

For the Transformer encoder, we start from a pretrained BERT model \citep{Turc:2019} which has 6 layers and 768 token dimension.
For DE and late-interaction models, we let the query encoder and document encoder share weights.
We use a query sequence length of 30 and a document sequence length of 200 with the Transformer.
If we use all 200 document tokens to calculate the similarity matrix $\mathbf{S}$, then $\mathbf{S} \in \mathbb{R}^{30 \times 200}$.
In some experiments the document sequence length is reduced in the end to save latency and storage cost; we will specify the details later.
More hyperparameter details are given in Appendix~\ref{app:hyperparams}.

\subsection{In-domain re-ranking on MS MARCO and NQ}

In \Cref{tab:msmarco_mrr}, we report MRR@10 and nDCG@10 scores for different scorers on all datasets.
When calculating the similarity matrix for ColBERT and LITE, we use the original sequence length (200) and token embedding dimension (768) of the Transformer encoder.
We try both the KL loss and margin MSE loss and report the better results; more details can be found in \Cref{sec:exp_losses}.

\begin{table*}[t]
    \centering
    \caption{MRR@10 and nDCG@10 scores. Separable LITE achieves the best in-domain results across all benchmarks.}
    \label{tab:msmarco_mrr}
    \begin{tabular}{lcccccccc}
        \toprule
         & \multicolumn{2}{c}{\textbf{MS MARCO}} & \multicolumn{2}{c}{\textbf{DL 2019}} & \multicolumn{2}{c}{\textbf{DL 2020}} & \multicolumn{2}{c}{\textbf{NQ}} \\
        \textbf{Scorer} & \textbf{MRR} & \textbf{nDCG} & \textbf{MRR} & \textbf{nDCG} & \textbf{MRR} & \textbf{nDCG} & \textbf{MRR} & \textbf{nDCG} \\
        \midrule
        DE & 0.355 & 0.413 & 0.861 & 0.744 & 0.842 & 0.723 & 0.699 & 0.611 \\
        ColBERT & 0.383 & 0.442 & 0.878 & 0.753 & 0.860 & 0.731 & 0.756 & 0.689 \\
        Sep LITE & \best{0.393} & \best{0.452} & \best{0.898} & \best{0.765} & \best{0.873} & \best{0.756} & \best{0.769} & \best{0.693} \\
        \bottomrule
    \end{tabular}    
\end{table*}

On MS MARCO, the {\tt T2} teacher~\citep{Hofstatter:2020} has Dev MRR@10 of 0.399.
A DE student can only achieve MRR@10 of 0.355.
Both ColBERT and separable LITE can significantly reduce this gap, but separable LITE is much better than ColBERT (0.393 vs. 0.383).
We also train a 6-layer, 768-dimensional CE student using distillation from the {\tt T2} teacher; it has MRR@10 of 0.395, which is only slightly better than separable LITE.
Moreover, on TREC DL 19 and 20 datasets, separable LITE also achieves better MRR@10 and nDCG@10 scores than ColBERT.

These observations generalize to the NQ dataset as well: we find that late-interaction models are much better than DE, and separable LITE is much better than ColBERT.

We also try a few ablations, including using top-$k$ aligned document tokens instead of top-$1$ in ColBERT, and freezing the backbone and only fine-tuning the scorers.
Separable LITE achieves better accuracy than ColBERT in all cases. See \Cref{sec:ablations} for details.

\subsection{Zero-shot re-ranking on BEIR}
\label{sec:beir}

\Cref{tab:beir_ndcg} presents zero-shot transfer results with ColBERT and separable LITE (from \Cref{tab:msmarco_mrr}) on 14 public datasets from BEIR~\citep{thakur2021beir}.
We also include results for the 6-layer CE model mentioned above, which is trained in the same way as other late-interaction models.
We can see that separable LITE achieves better zero-shot transfer than ColBERT on 11 out of 14 datasets.
CE still gives better zero-shot transfer than separable LITE, but as we show below, CE has much higher latency (cf. \Cref{tab:latency}).

\begin{table}[!t]
    \centering
    \caption{BEIR nDCG@10. Separable LITE is better than ColBERT on 11 out of 14 datasets.}
    \label{tab:beir_ndcg}
    \begin{tabular}{@{}lcc|c}
        \toprule
        \textbf{Dataset} & \textbf{ColBERT} & \textbf{Sep LITE} & \textbf{CE} \\
        \midrule
        T-COVID & 0.761 & \best{0.763} & 0.771 \\
        NFCorpus & 0.356 & \best{0.358} & 0.361 \\
        NQ & 0.525 & \best{0.540} & 0.552 \\
        HotpotQA & \best{0.685} & 0.681 & 0.728 \\
        FiQA-2018 & 0.330 & \best{0.336} & 0.346 \\
        ArguAna & \best{0.433} & 0.424 & 0.519 \\
        Touch\'{e}-2020 & 0.274 & \best{0.305} & 0.300 \\
        CQAD & 0.363 & \best{0.374} & 0.378 \\
        Quora & 0.767 & \best{0.839} & 0.832 \\
        DBPedia & 0.410 & \best{0.434} & 0.438 \\
        SCIDOCS & 0.155 & \best{0.164} & 0.167 \\
        FEVER & 0.782 & \best{0.788} & 0.804 \\
        C-FEVER & 0.190 & \best{0.213} & 0.232 \\
        SciFact & \best{0.667} & 0.633 & 0.695 \\
        \bottomrule
    \end{tabular}    
\end{table}

\subsection{Results on MS MARCO with reduced latency and storage}
\label{sec:limited_storage}

As discussed previously, late-interaction methods may have higher latency and storage cost than DE. 
Suppose the Transformer encoders use $L_1$ query tokens and $L_2$ document tokens of dimension $P$, then DE only needs to take one dot product, while calculating the similarity matrix for late-interaction methods requires $L_1L_2$ dot products. 
Moreover, to save online latency, we need to pre-compute and store one $P$-dimensional document embedding vector for DE, while for late-interaction methods we might need to store a $P\times L_2$ embedding matrix.
This increase in storage cost is significant in industry-scale information retrieval systems, since there can be billions of documents \citep{zhang2013image,overwijk2022clueweb22}.

One solution is to reduce $P$ and $L_2$ to some smaller $P'$ and $L_2'$ (by projection, pooling, etc.), and then store a $P'\times L_2'$ embedding matrix for each document.
Correspondingly, for each query we use $L_1$ embedding vectors of dimension $P'$, and to calculate the similarity matrix, we need $L_1L_2'$ dot products between $P'$-dimensional vectors.
This can reduce both latency and storage; below we analyze how performance drops with such reduction, and show that separable LITE remains competitive compared to ColBERT.

\paragraph{Reducing the number of output document tokens.}

\begin{figure}[!t]
    \begin{center}
        \includegraphics[width=0.45\textwidth]{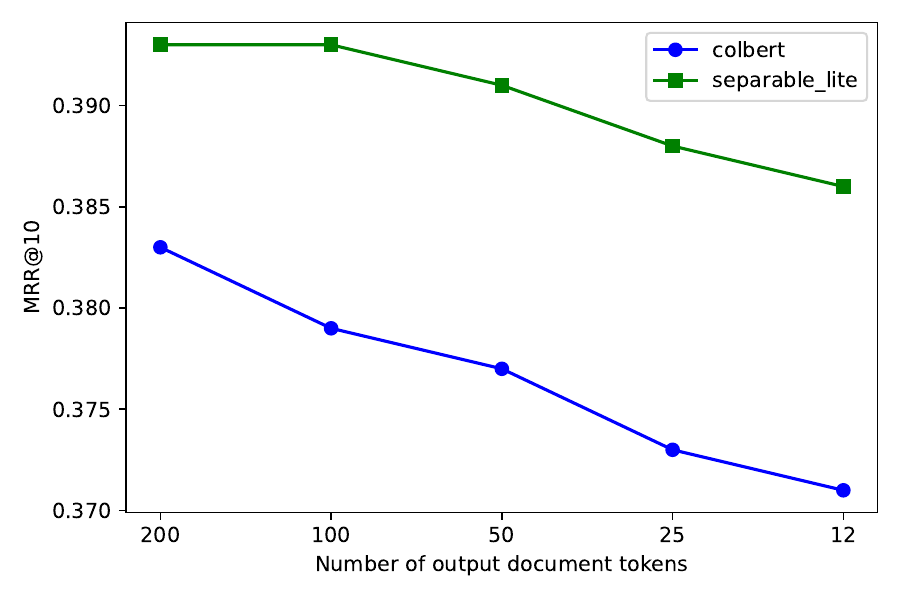}
        \caption{MS MARCO MRR with fewer document tokens.}    
        \label{fig:reduce_L2}
    \end{center}
\end{figure}

Here, we keep the token dimension at 768 and reduce the number of output document tokens.
The Transformer encoder outputs an embedding matrix $\mathbf{D}\in\R^{768\times200}$ of 200 token embeddings, and we try to reduce the number of tokens either by directly taking average of adjacent columns (average pooling), or by applying a trainable linear projection to every row of $\mathbf{D}$.
We try both methods and find that separable LITE prefers learnable projection while ColBERT prefers average pooling.
The results are shown in \Cref{fig:reduce_L2}, and we can see separable LITE is more accurate than ColBERT with reduced document sequence lengths.

\paragraph{Reducing token dimension.}

\begin{figure}[!t]
    \begin{center}
        \includegraphics[width=0.45\textwidth]{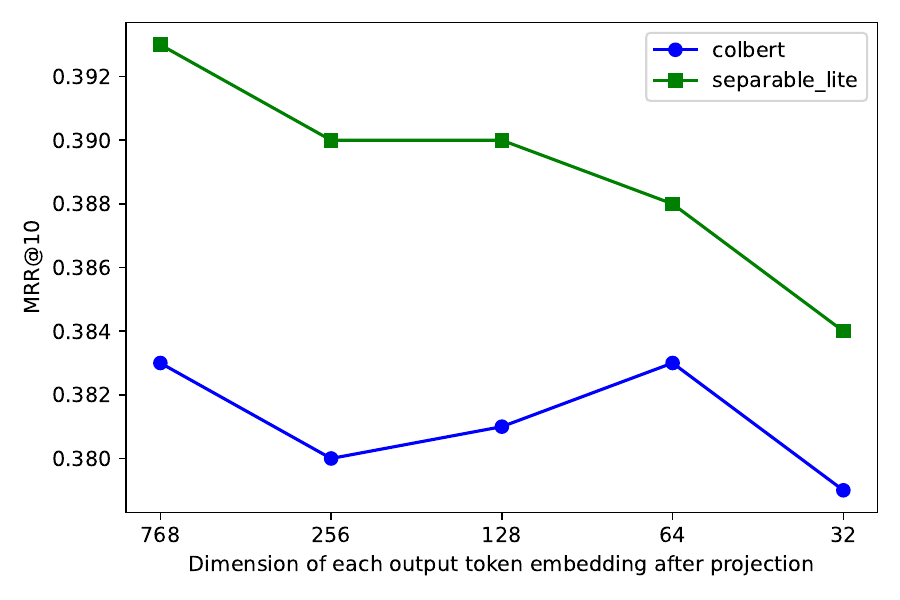}
        \caption{MS MARCO MRR with reduced token dimension.}
        \label{fig:reduce_token_dim}
    \end{center}
\end{figure}

Next, we fix the number of document tokens at 200, and reduce the dimension of each output token via learnable linear projections.
The results are given in \Cref{fig:reduce_token_dim}.
With different token dimension, separable LITE is always more accurate than ColBERT.

\paragraph{Achieving lower latency/storage than ColBERT using LITE.}
\label{sec:latency}

If the size of pre-computed document embedding matrix is fixed, then LITE has higher latency than ColBERT since its MLP scorer is slower than sum-max.
However, since LITE is robust to embedding size reduction, it can remain more accurate than ColBERT while being more time and space efficient by using fewer document tokens. 
The result is shown in \Cref{tab:latency}, together with latency of other scorers studied before.

In \Cref{tab:latency}, we evaluate the latency of scoring relevance between 1 query and 100 documents.
For CE, we use the 6-layer distilled student and evaluate the total time to calculate the joint embeddings between the query and every document.
For DE, ColBERT and separable LITE, we use models from \Cref{tab:msmarco_mrr}; we pre-compute the document embeddings, and evaluate the query embedding generation and scoring time.
For the ``small separable LITE'' model, we only store 50 tokens for each document, and we also use a small MLP (we let $\mathbf{W}_1$ in~\eqref{eq:sep_lite_row} have shape $(768,50)$).
In \Cref{tab:latency}, small separable LITE only uses $0.25\times$ storage space compared with ColBERT which stores 200 document token embeddings, and it also achieves lower latency while still being much more accurate than ColBERT (0.391 vs. 0.383).
In \Cref{tab:beir_ndcg_small_sep}, we show that small separate LITE is better than ColBERT on 8 out of 14 datasets.
We can also see that the CE latency is $100\times$ of the LITE latency, since CE cannot use offline pre-computation.

\begin{table}[!t]
    \centering
    \caption{Latency of different scorers.}
    \label{tab:latency}
    \vskip 0.1in
    \begin{tabular}{@{}lrrr}
        \toprule
        \textbf{Scorer} & Latency & Storage & MS MARCO \\
         & (in ms) & & MRR@10 \\
        \midrule
        CE (student) & 10990  & 0$\times$ & 0.395  \\
        DE & 42  & 1$\times$ & 0.355  \\
        ColBERT & 62 & 200$\times$ & 0.383 \\
        Separable LITE & 111 & 200$\times$ & 0.393  \\
        Small sep LITE & 56 & 50$\times$ & 0.391 \\
        \bottomrule
    \end{tabular}
\end{table}

\subsection{Comparison with KNRM}

KNRM~\citep{Xiong:2017} is one popular pre-Transformer scorer; it calculates the similarity matrix using Word2Vec embeddings, and then apply kernel pooling.
It has been applied to MS MARCO in a few recent works \citep{Khattab:2020,Hofstatter:2020}; however, KNRM only achieves low accuracy, likely because the underlying encoders are non-pretrained shallow Transformers.
In this work, we try to apply KNRM with the same pretrained BERT encoder as other scorers.
We find that KNRM can achieve similar accuracy to ColBERT overall, but separable LITE is still better than KNRM over all benchmarks; see \Cref{sec:knrm_results} for details.

\section{Conclusion}

In this work, we propose LITE models that can provably approximate any continuous scoring functions.
We also show that LITE outperforms prior late-interaction models in both in-domain and zero-shot reranking.
In particular, LITE can achieve higher accuracy with less latency and storage cost.

\section*{Limitations}

In our MS MARCO experiments, we only train our models using triplet data; by contrast, state-of-the-art models such as ColBERTv2 \citep{Santhanam:2022} use additional techniques such as hard-negative mining.
One next step is to evaluate LITE with these techniques.
Additionally, our proposed LITE model is suitable for the re-ranking phase of information retrieval.
However, given that it is built on top of a factorized dual-encoder,
can one also adapt it for use in the retrieval phase? For instance, one possibility could be to jointly train retrieval embeddings and the LITE model such that both the models share the same encoders.
Such an analysis is also important and needed in future work.

\section*{Ethics Statement}

LITE is a general technique that can improve relevance scoring accuracy compared with simple operations such as dot products, and we do not see potential risks. In particular, LITE is only a scoring module and does not generate harmful information. We do need to train the LITE scorer and fine tune the underlying Transformer encoder, which could have some environmental effect.

\bibliographystyle{plainnat}
\bibliography{references}

\begin{thebibliography}{51}
\providecommand{\natexlab}[1]{#1}
\providecommand{\url}[1]{\texttt{#1}}
\expandafter\ifx\csname urlstyle\endcsname\relax
  \providecommand{\doi}[1]{doi: #1}\else
  \providecommand{\doi}{doi: \begingroup \urlstyle{rm}\Url}\fi

\bibitem[Chollet(2017)]{xception}
Fran{\c{c}}ois Chollet.
\newblock Xception: Deep learning with depthwise separable convolutions.
\newblock In \emph{Proceedings of the IEEE conference on computer vision and
  pattern recognition}, pages 1251--1258, 2017.

\bibitem[Craswell et~al.(2020)Craswell, Mitra, Yilmaz, Campos, and
  Voorhees]{trec_19}
Nick Craswell, Bhaskar Mitra, Emine Yilmaz, Daniel Campos, and Ellen~M
  Voorhees.
\newblock Overview of the trec 2019 deep learning track.
\newblock \emph{arXiv preprint arXiv:2003.07820}, 2020.

\bibitem[Craswell et~al.(2021)Craswell, Mitra, Yilmaz, and Campos]{trec_20}
Nick Craswell, Bhaskar Mitra, Emine Yilmaz, and Daniel Campos.
\newblock Overview of the trec 2020 deep learning track.
\newblock \emph{arXiv preprint arXiv:2102.07662}, 2021.

\bibitem[Cybenko(1989)]{cybenko1989approximation}
George Cybenko.
\newblock Approximation by superpositions of a sigmoidal function.
\newblock \emph{Mathematics of control, signals and systems}, 2\penalty0
  (4):\penalty0 303--314, 1989.

\bibitem[Dai et~al.(2018)Dai, Xiong, Callan, and Liu]{Dai:2018ConvKNRM}
Zhuyun Dai, Chenyan Xiong, Jamie Callan, and Zhiyuan Liu.
\newblock Convolutional neural networks for soft-matching n-grams in ad-hoc
  search.
\newblock In \emph{Proceedings of the Eleventh ACM International Conference on
  Web Search and Data Mining}, WSDM '18, page 126–134, New York, NY, USA,
  2018. Association for Computing Machinery.
\newblock ISBN 9781450355810.
\newblock \doi{10.1145/3159652.3159659}.
\newblock URL \url{https://doi.org/10.1145/3159652.3159659}.

\bibitem[Devlin et~al.(2019)Devlin, Chang, Lee, and Toutanova]{Devlin:2019}
Jacob Devlin, Ming{-}Wei Chang, Kenton Lee, and Kristina Toutanova.
\newblock {BERT:} pre-training of deep bidirectional transformers for language
  understanding.
\newblock In Jill Burstein, Christy Doran, and Thamar Solorio, editors,
  \emph{Proceedings of the 2019 Conference of the North American Chapter of the
  Association for Computational Linguistics: Human Language Technologies,
  {NAACL-HLT} 2019, Minneapolis, MN, USA, June 2-7, 2019, Volume 1 (Long and
  Short Papers)}, pages 4171--4186. Association for Computational Linguistics,
  2019.

\bibitem[Dziugaite and Roy(2015)]{Dziugaite:2015}
Gintare~Karolina Dziugaite and Daniel~M. Roy.
\newblock Neural network matrix factorization.
\newblock \emph{CoRR}, abs/1511.06443, 2015.
\newblock URL \url{http://arxiv.org/abs/1511.06443}.

\bibitem[Funahashi(1989)]{funahashi1989approximate}
Ken-Ichi Funahashi.
\newblock On the approximate realization of continuous mappings by neural
  networks.
\newblock \emph{Neural networks}, 2\penalty0 (3):\penalty0 183--192, 1989.

\bibitem[Gao et~al.(2021)Gao, Dai, and Callan]{gao2021coil}
Luyu Gao, Zhuyun Dai, and Jamie Callan.
\newblock Coil: Revisit exact lexical match in information retrieval with
  contextualized inverted list.
\newblock \emph{arXiv preprint arXiv:2104.07186}, 2021.

\bibitem[Guo et~al.(2016{\natexlab{a}})Guo, Fan, Ai, and Croft]{Guo:2016}
Jiafeng Guo, Yixing Fan, Qingyao Ai, and W.~Bruce Croft.
\newblock A deep relevance matching model for ad-hoc retrieval.
\newblock In \emph{Proceedings of the 25th ACM International on Conference on
  Information and Knowledge Management}, CIKM '16, page 55–64, New York, NY,
  USA, 2016{\natexlab{a}}. Association for Computing Machinery.
\newblock ISBN 9781450340731.

\bibitem[Guo et~al.(2016{\natexlab{b}})Guo, Kumar, Choromanski, and
  Simcha]{guo2016quantization}
Ruiqi Guo, Sanjiv Kumar, Krzysztof Choromanski, and David Simcha.
\newblock Quantization based fast inner product search.
\newblock In \emph{Artificial intelligence and statistics}, pages 482--490.
  PMLR, 2016{\natexlab{b}}.

\bibitem[Guo et~al.(2020)Guo, Sun, Lindgren, Geng, Simcha, Chern, and
  Kumar]{Guo:2020}
Ruiqi Guo, Philip Sun, Erik Lindgren, Quan Geng, David Simcha, Felix Chern, and
  Sanjiv Kumar.
\newblock Accelerating large-scale inference with anisotropic vector
  quantization.
\newblock In \emph{Proceedings of the 37th International Conference on Machine
  Learning, {ICML} 2020, 13-18 July 2020, Virtual Event}, volume 119 of
  \emph{Proceedings of Machine Learning Research}, pages 3887--3896. {PMLR},
  2020.

\bibitem[He et~al.(2017)He, Liao, Zhang, Nie, Hu, and Chua]{he2017neural}
Xiangnan He, Lizi Liao, Hanwang Zhang, Liqiang Nie, Xia Hu, and Tat-Seng Chua.
\newblock Neural collaborative filtering.
\newblock In \emph{Proceedings of the 26th international conference on world
  wide web}, pages 173--182, 2017.

\bibitem[Hofst{\"{a}}tter et~al.(2020)Hofst{\"{a}}tter, Althammer,
  Schr{\"{o}}der, Sertkan, and Hanbury]{Hofstatter:2020}
Sebastian Hofst{\"{a}}tter, Sophia Althammer, Michael Schr{\"{o}}der, Mete
  Sertkan, and Allan Hanbury.
\newblock Improving efficient neural ranking models with cross-architecture
  knowledge distillation.
\newblock \emph{CoRR}, abs/2010.02666, 2020.
\newblock URL \url{https://arxiv.org/abs/2010.02666}.

\bibitem[Hornik et~al.(1989)Hornik, Stinchcombe, and
  White]{hornik1989multilayer}
Kurt Hornik, Maxwell Stinchcombe, and Halbert White.
\newblock Multilayer feedforward networks are universal approximators.
\newblock \emph{Neural networks}, 2\penalty0 (5):\penalty0 359--366, 1989.

\bibitem[J{\"a}rvelin and Kek{\"a}l{\"a}inen(2002)]{jarvelin2002cumulated}
Kalervo J{\"a}rvelin and Jaana Kek{\"a}l{\"a}inen.
\newblock Cumulated gain-based evaluation of ir techniques.
\newblock \emph{ACM Transactions on Information Systems (TOIS)}, 20\penalty0
  (4):\penalty0 422--446, 2002.

\bibitem[Johnson et~al.(2019)Johnson, Douze, and J{\'e}gou]{johnson2019billion}
Jeff Johnson, Matthijs Douze, and Herv{\'e} J{\'e}gou.
\newblock Billion-scale similarity search with {GPUs}.
\newblock \emph{IEEE Transactions on Big Data}, 7\penalty0 (3):\penalty0
  535--547, 2019.

\bibitem[Karpukhin et~al.(2020)Karpukhin, Oguz, Min, Lewis, Wu, Edunov, Chen,
  and Yih]{Karpukhin:2020}
Vladimir Karpukhin, Barlas Oguz, Sewon Min, Patrick Lewis, Ledell Wu, Sergey
  Edunov, Danqi Chen, and Wen-tau Yih.
\newblock Dense passage retrieval for open-domain question answering.
\newblock In \emph{Proceedings of the 2020 Conference on Empirical Methods in
  Natural Language Processing (EMNLP)}, pages 6769--6781, Online, November
  2020. Association for Computational Linguistics.

\bibitem[Khattab and Zaharia(2020)]{Khattab:2020}
Omar Khattab and Matei Zaharia.
\newblock \emph{ColBERT: Efficient and Effective Passage Search via
  Contextualized Late Interaction over BERT}, page 39–48.
\newblock Association for Computing Machinery, New York, NY, USA, 2020.
\newblock ISBN 9781450380164.

\bibitem[Kudo and Richardson(2018)]{Kudo:2018}
Taku Kudo and John Richardson.
\newblock {S}entence{P}iece: A simple and language independent subword
  tokenizer and detokenizer for neural text processing.
\newblock In \emph{Proceedings of the 2018 Conference on Empirical Methods in
  Natural Language Processing: System Demonstrations}, pages 66--71, Brussels,
  Belgium, November 2018. Association for Computational Linguistics.
\newblock \doi{10.18653/v1/D18-2012}.
\newblock URL \url{https://aclanthology.org/D18-2012}.

\bibitem[Kwiatkowski et~al.(2019)Kwiatkowski, Palomaki, Redfield, Collins,
  Parikh, Alberti, Epstein, Polosukhin, Kelcey, Devlin, Lee, Toutanova, Jones,
  Chang, Dai, Uszkoreit, Le, and Petrov]{Kwiatkowski:2019}
Tom Kwiatkowski, Jennimaria Palomaki, Olivia Redfield, Michael Collins, Ankur
  Parikh, Chris Alberti, Danielle Epstein, Illia Polosukhin, Matthew Kelcey,
  Jacob Devlin, Kenton Lee, Kristina~N. Toutanova, Llion Jones, Ming-Wei Chang,
  Andrew Dai, Jakob Uszkoreit, Quoc Le, and Slav Petrov.
\newblock Natural questions: a benchmark for question answering research.
\newblock \emph{Transactions of the Association of Computational Linguistics},
  2019.

\bibitem[Li et~al.(2022)Li, Lin, Oguz, Ghoshal, Lin, Mehdad, Yih, and
  Chen]{li2022citadel}
Minghan Li, Sheng-Chieh Lin, Barlas Oguz, Asish Ghoshal, Jimmy Lin, Yashar
  Mehdad, Wen-tau Yih, and Xilun Chen.
\newblock Citadel: Conditional token interaction via dynamic lexical routing
  for efficient and effective multi-vector retrieval.
\newblock \emph{arXiv preprint arXiv:2211.10411}, 2022.

\bibitem[Li et~al.(2023)Li, Lin, Ma, and Lin]{li2023slim}
Minghan Li, Sheng-Chieh Lin, Xueguang Ma, and Jimmy Lin.
\newblock Slim: Sparsified late interaction for multi-vector retrieval with
  inverted indexes.
\newblock \emph{arXiv preprint arXiv:2302.06587}, 2023.

\bibitem[Loshchilov and Hutter(2019)]{loshchilov2018decoupled}
Ilya Loshchilov and Frank Hutter.
\newblock Decoupled weight decay regularization.
\newblock In \emph{International Conference on Learning Representations}, 2019.
\newblock URL \url{https://openreview.net/forum?id=Bkg6RiCqY7}.

\bibitem[Luan et~al.(2021)Luan, Eisenstein, Toutanova, and Collins]{Yi:2021}
Yi~Luan, Jacob Eisenstein, Kristina Toutanova, and Michael Collins.
\newblock Sparse, dense, and attentional representations for text retrieval.
\newblock \emph{Transactions of the Association for Computational Linguistics},
  9:\penalty0 329--345, 2021.

\bibitem[Ma et~al.(2023)Ma, Wang, Yang, Wei, and Lin]{ma2023fine}
Xueguang Ma, Liang Wang, Nan Yang, Furu Wei, and Jimmy Lin.
\newblock Fine-tuning llama for multi-stage text retrieval.
\newblock \emph{arXiv preprint arXiv:2310.08319}, 2023.

\bibitem[MacAvaney et~al.(2019)MacAvaney, Yates, Cohan, and
  Goharian]{MacAvaney:2019}
Sean MacAvaney, Andrew Yates, Arman Cohan, and Nazli Goharian.
\newblock {CEDR}: Contextualized embeddings for document ranking.
\newblock In \emph{SIGIR}, 2019.

\bibitem[Menon et~al.(2022)Menon, Jayasumana, Rawat, Kim, Reddi, and
  Kumar]{Menon:2022DE}
Aditya Menon, Sadeep Jayasumana, Ankit~Singh Rawat, Seungyeon Kim, Sashank
  Reddi, and Sanjiv Kumar.
\newblock In defense of dual-encoders for neural ranking.
\newblock In Kamalika Chaudhuri, Stefanie Jegelka, Le~Song, Csaba Szepesvari,
  Gang Niu, and Sivan Sabato, editors, \emph{Proceedings of the 39th
  International Conference on Machine Learning}, volume 162 of
  \emph{Proceedings of Machine Learning Research}, pages 15376--15400. PMLR,
  17--23 Jul 2022.
\newblock URL \url{https://proceedings.mlr.press/v162/menon22a.html}.

\bibitem[Mitra and Craswell(2018)]{Mitra:2018}
Bhaskar Mitra and Nick Craswell.
\newblock An introduction to neural information retrieval.
\newblock \emph{Foundations and Trends{\textregistered} in Information
  Retrieval}, 13\penalty0 (1):\penalty0 1--126, 2018.
\newblock ISSN 1554-0669.
\newblock \doi{10.1561/1500000061}.

\bibitem[Mysore et~al.(2021)Mysore, Cohan, and Hope]{mysore2021multi}
Sheshera Mysore, Arman Cohan, and Tom Hope.
\newblock Multi-vector models with textual guidance for fine-grained scientific
  document similarity.
\newblock \emph{arXiv preprint arXiv:2111.08366}, 2021.

\bibitem[Nguyen et~al.(2016)Nguyen, Rosenberg, Song, Gao, Tiwary, Majumder, and
  Deng]{Nguyen:2016}
Tri Nguyen, Mir Rosenberg, Xia Song, Jianfeng Gao, Saurabh Tiwary, Rangan
  Majumder, and Li~Deng.
\newblock {MS} {MARCO:} {A} human generated machine reading comprehension
  dataset.
\newblock In Tarek~Richard Besold, Antoine Bordes, Artur~S. d'Avila Garcez, and
  Greg Wayne, editors, \emph{Proceedings of the Workshop on Cognitive
  Computation: Integrating neural and symbolic approaches 2016}, volume 1773 of
  \emph{{CEUR} Workshop Proceedings}. CEUR-WS.org, 2016.

\bibitem[Ni et~al.(2021)Ni, Qu, Lu, Dai, {\'{A}}brego, Ma, Zhao, Luan, Hall,
  Chang, and Yang]{t5_t3}
Jianmo Ni, Chen Qu, Jing Lu, Zhuyun Dai, Gustavo~Hern{\'{a}}ndez {\'{A}}brego,
  Ji~Ma, Vincent~Y. Zhao, Yi~Luan, Keith~B. Hall, Ming{-}Wei Chang, and Yinfei
  Yang.
\newblock Large dual encoders are generalizable retrievers.
\newblock \emph{CoRR}, abs/2112.07899, 2021.
\newblock URL \url{https://arxiv.org/abs/2112.07899}.

\bibitem[Nogueira and Cho(2019)]{Nogueira:2019}
Rodrigo Nogueira and Kyunghyun Cho.
\newblock Passage re-ranking with {BERT}.
\newblock \emph{CoRR}, abs/1901.04085, 2019.
\newblock URL \url{http://arxiv.org/abs/1901.04085}.

\bibitem[Overwijk et~al.(2022)Overwijk, Xiong, and
  Callan]{overwijk2022clueweb22}
Arnold Overwijk, Chenyan Xiong, and Jamie Callan.
\newblock Clueweb22: 10 billion web documents with rich information.
\newblock In \emph{Proceedings of the 45th International ACM SIGIR Conference
  on Research and Development in Information Retrieval}, pages 3360--3362,
  2022.

\bibitem[Pang et~al.(2016)Pang, Lan, Guo, Xu, Wan, and Cheng]{match_pyramid}
Liang Pang, Yanyan Lan, Jiafeng Guo, Jun Xu, Shengxian Wan, and Xueqi Cheng.
\newblock Text matching as image recognition.
\newblock In \emph{Proceedings of the AAAI Conference on Artificial
  Intelligence}, volume~30, 2016.

\bibitem[Qian et~al.(2022)Qian, Lee, Duddu, Dai, Brahma, Naim, Lei, and
  Zhao]{qian2022multi}
Yujie Qian, Jinhyuk Lee, Sai Meher~Karthik Duddu, Zhuyun Dai, Siddhartha
  Brahma, Iftekhar Naim, Tao Lei, and Vincent~Y Zhao.
\newblock Multi-vector retrieval as sparse alignment.
\newblock \emph{arXiv preprint arXiv:2211.01267}, 2022.

\bibitem[Qu et~al.(2021)Qu, Ding, Liu, Liu, Ren, Zhao, Dong, Wu, and
  Wang]{Qu:2021}
Yingqi Qu, Yuchen Ding, Jing Liu, Kai Liu, Ruiyang Ren, Wayne~Xin Zhao, Daxiang
  Dong, Hua Wu, and Haifeng Wang.
\newblock Rocket{QA}: An optimized training approach to dense passage retrieval
  for open-domain question answering.
\newblock In Kristina Toutanova, Anna Rumshisky, Luke Zettlemoyer, Dilek
  Hakkani{-}T{\"{u}}r, Iz~Beltagy, Steven Bethard, Ryan Cotterell, Tanmoy
  Chakraborty, and Yichao Zhou, editors, \emph{Proceedings of the 2021
  Conference of the North American Chapter of the Association for Computational
  Linguistics: Human Language Technologies, {NAACL-HLT} 2021, Online, June
  6-11, 2021}, pages 5835--5847. Association for Computational Linguistics,
  2021.

\bibitem[Radev et~al.(2002)Radev, Qi, Wu, and Fan]{Radev:2002}
Dragomir~R. Radev, Hong Qi, Harris Wu, and Weiguo Fan.
\newblock Evaluating web-based question answering systems.
\newblock In \emph{Proceedings of the Third International Conference on
  Language Resources and Evaluation ({LREC}{'}02)}, Las Palmas, Canary Islands
  - Spain, May 2002. European Language Resources Association (ELRA).

\bibitem[Reimers and Gurevych(2019)]{Reimers:2019}
Nils Reimers and Iryna Gurevych.
\newblock Sentence-bert: Sentence embeddings using {S}iamese {BERT}-networks.
\newblock In \emph{Proceedings of the 2019 Conference on Empirical Methods in
  Natural Language Processing}. Association for Computational Linguistics, 11
  2019.
\newblock URL \url{https://arxiv.org/abs/1908.10084}.

\bibitem[Ren et~al.(2021)Ren, Qu, Liu, Zhao, She, Wu, Wang, and Wen]{Ren:2021}
Ruiyang Ren, Yingqi Qu, Jing Liu, Wayne~Xin Zhao, QiaoQiao She, Hua Wu, Haifeng
  Wang, and Ji-Rong Wen.
\newblock {R}ocket{QA}v2: A joint training method for dense passage retrieval
  and passage re-ranking.
\newblock In \emph{Proceedings of the 2021 Conference on Empirical Methods in
  Natural Language Processing}, pages 2825--2835, Online and Punta Cana,
  Dominican Republic, November 2021. Association for Computational Linguistics.

\bibitem[Rendle et~al.(2020)Rendle, Krichene, Zhang, and
  Anderson]{rendle2020neural}
Steffen Rendle, Walid Krichene, Li~Zhang, and John Anderson.
\newblock Neural collaborative filtering vs. matrix factorization revisited.
\newblock In \emph{Proceedings of the 14th ACM Conference on Recommender
  Systems}, pages 240--248, 2020.

\bibitem[Robertson and Zaragoza(2009)]{Robertson:2009}
Stephen Robertson and Hugo Zaragoza.
\newblock The probabilistic relevance framework: Bm25 and beyond.
\newblock \emph{Found. Trends Inf. Retr.}, 3\penalty0 (4):\penalty0 333–389,
  April 2009.
\newblock ISSN 1554-0669.

\bibitem[Santhanam et~al.(2022)Santhanam, Khattab, Saad-Falcon, Potts, and
  Zaharia]{Santhanam:2022}
Keshav Santhanam, Omar Khattab, Jon Saad-Falcon, Christopher Potts, and Matei
  Zaharia.
\newblock {C}ol{BERT}v2: Effective and efficient retrieval via lightweight late
  interaction.
\newblock In \emph{Proceedings of the 2022 Conference of the North American
  Chapter of the Association for Computational Linguistics: Human Language
  Technologies}, pages 3715--3734, Seattle, United States, July 2022.
  Association for Computational Linguistics.
\newblock \doi{10.18653/v1/2022.naacl-main.272}.
\newblock URL \url{https://aclanthology.org/2022.naacl-main.272}.

\bibitem[Thakur et~al.(2021)Thakur, Reimers, R{\"u}ckl{\'e}, Srivastava, and
  Gurevych]{thakur2021beir}
Nandan Thakur, Nils Reimers, Andreas R{\"u}ckl{\'e}, Abhishek Srivastava, and
  Iryna Gurevych.
\newblock {BEIR}: A heterogenous benchmark for zero-shot evaluation of
  information retrieval models.
\newblock \emph{arXiv preprint arXiv:2104.08663}, 2021.

\bibitem[Tolstikhin et~al.(2021)Tolstikhin, Houlsby, Kolesnikov, Beyer, Zhai,
  Unterthiner, Yung, Steiner, Keysers, Uszkoreit, et~al.]{mlp_mixer}
Ilya~O Tolstikhin, Neil Houlsby, Alexander Kolesnikov, Lucas Beyer, Xiaohua
  Zhai, Thomas Unterthiner, Jessica Yung, Andreas Steiner, Daniel Keysers,
  Jakob Uszkoreit, et~al.
\newblock {MLP}-mixer: An all-mlp architecture for vision.
\newblock \emph{Advances in Neural Information Processing Systems},
  34:\penalty0 24261--24272, 2021.

\bibitem[Turc et~al.(2019)Turc, Chang, Lee, and Toutanova]{Turc:2019}
Iulia Turc, Ming{-}Wei Chang, Kenton Lee, and Kristina Toutanova.
\newblock Well-read students learn better: The impact of student initialization
  on knowledge distillation.
\newblock \emph{CoRR}, abs/1908.08962, 2019.
\newblock URL \url{http://arxiv.org/abs/1908.08962}.

\bibitem[Vaswani et~al.(2017)Vaswani, Shazeer, Parmar, Uszkoreit, Jones, Gomez,
  Kaiser, and Polosukhin]{Vaswani:2017}
Ashish Vaswani, Noam Shazeer, Niki Parmar, Jakob Uszkoreit, Llion Jones,
  Aidan~N. Gomez, \L{}ukasz Kaiser, and Illia Polosukhin.
\newblock Attention is all you need.
\newblock In \emph{Proceedings of the 31st International Conference on Neural
  Information Processing Systems}, NIPS'17, page 6000–6010, Red Hook, NY,
  USA, 2017. Curran Associates Inc.
\newblock ISBN 9781510860964.

\bibitem[Xiong et~al.(2017)Xiong, Dai, Callan, Liu, and Power]{Xiong:2017}
Chenyan Xiong, Zhuyun Dai, Jamie Callan, Zhiyuan Liu, and Russell Power.
\newblock End-to-end neural ad-hoc ranking with kernel pooling.
\newblock In \emph{Proceedings of the 40th International ACM SIGIR Conference
  on Research and Development in Information Retrieval}, SIGIR '17, page
  55–64, New York, NY, USA, 2017. Association for Computing Machinery.
\newblock ISBN 9781450350228.

\bibitem[Yun et~al.(2020)Yun, Bhojanapalli, Rawat, Reddi, and Kumar]{Yun:2020}
Chulhee Yun, Srinadh Bhojanapalli, Ankit~Singh Rawat, Sashank Reddi, and Sanjiv
  Kumar.
\newblock Are transformers universal approximators of sequence-to-sequence
  functions?
\newblock In \emph{International Conference on Learning Representations}, 2020.

\bibitem[Zhang and Rui(2013)]{zhang2013image}
Lei Zhang and Yong Rui.
\newblock Image search—from thousands to billions in 20 years.
\newblock \emph{ACM Transactions on Multimedia Computing, Communications, and
  Applications (TOMM)}, 9\penalty0 (1s):\penalty0 1--20, 2013.

\bibitem[Zhu et~al.(2023)Zhu, Lin, Anand, Calderwood, Clausen-Brown, Lueck,
  Yim, and Wu]{zhu2023explicit}
Xiaofeng Zhu, Thomas Lin, Vishal Anand, Matthew Calderwood, Eric Clausen-Brown,
  Gord Lueck, Wen-wai Yim, and Cheng Wu.
\newblock Explicit and implicit semantic ranking framework.
\newblock In \emph{Companion Proceedings of the ACM Web Conference 2023}, pages
  326--330, 2023.

\end{thebibliography}

\appendix

\newpage
\appendix
\onecolumn

\section{Experimental details}
\label{sec:exp_detail}

\subsection{Hyper-parameters}
\label{app:hyperparams}

The main hyperparameters for LITE are the MLP widths.
For Separable LITE (cf.~\eqref{eq:sep_lite_row} and \eqref{eq:sep_lite_col}), if the input dot-product matrix has shape $L_1\times L_2$, then $\mathbf{W}_1$ has shape $(m_2,L_2)$, $\mathbf{W}_2$ has shape $(L_2,m_2)$, $\mathbf{W}_3$ has shape $(m_1,L_1)$, and $\mathbf{W}_4$ has shape $(L_1,m_1)$.
In this work, we let $m_1=360$ and $m_2=2400$ in most experiments for simplicity, but we also note that much smaller widths can already give a high accuracy while also reducing the latency (cf. \Cref{tab:latency}).

\subsection{Training details}
\label{sec:loss_metric}

Here we first define the loss functions used in our experiments. 

For simplicity, let us first consider the triplet setting, where we are given a query $q$, a positive document $d_+$, and a negative document $d_-$. 
Suppose the teacher score is given by $\mathbf{t}=(t_+,t_-)$, and the student score is $\mathbf{s}=(s_+,s_-)$. 
The margin MSE loss is defined as $\left((t_+-t_-)-(s_+-s_-)\right)^2$, i.e., it calculates the teacher score margin and student score margin, and applies a squared loss.
The KL loss first calculates the teacher and student probability distributions as below
\begin{align*}
    \mathbf{p}^{(t)} & =\left(\frac{\exp(t_+)}{\exp(t_+)+\exp(t_-)},\frac{\exp(t_-)}{\exp(t_+)+\exp(t_-)}\right), \\
    \mathbf{p}^{(s)} & =\left(\frac{\exp(s_+)}{\exp(s_+)+\exp(s_-)},\frac{\exp(s_-)}{\exp(s_+)+\exp(s_-)}\right),
\end{align*}
and then calculates the KL divergence ${\rm KL}(\mathbf{p}^{(t)}||\mathbf{p}^{(s)})$.

In our NQ experiments, we use one positive document and multiple negative documents. In this case the KL loss is defined similarly, while for the margin MSE loss we consider the margins between the positive document and every negative document.
Formally, suppose there are $N$ documents, the first one is positive while the remaining ones are negative, and let $t_i$ and $s_i$ denote the teacher and student scores for the $i$-th document, then we consider 
\begin{equation*}
    \sum_{i=2}^N((t_1-t_i)-(s_1-s_i))^2.
\end{equation*}

It is also an interesting open direction to try other training frameworks, such as sRank \citep{zhu2023explicit}.

On the optimization algorithm, we use AdamW~\citep{loshchilov2018decoupled} with batch size 128, peak learning rate $2.8 \times 10^{-5}$, weight decay 0.01, and 1.5 million steps.
We use a linear learning rate warm up of 30000 steps, then a linear learning rate decay.

\subsection{Results with different loss functions}
\label{sec:exp_losses}

Here we present results on different scorers and loss functions.

First, \Cref{tab:msmarco_mrr_loss} includes results on MS MARCO.

\begin{table}[h]
    \centering
    \caption{MS MARCO Dev MRR@10. Separable LITE achieves the best results among factorized (non-CE) models.}
    \label{tab:msmarco_mrr_loss}
    \vskip 0.1in
    \scalebox{0.95}{
    \begin{tabular}{@{}lcc}
        \toprule
        \textbf{Scorer} & \textbf{KL} & \textbf{Margin MSE} \\
        \midrule
        CE student & 0.394 & 0.395 \\
        \midrule
        DE & 0.355 & 0.350 \\
        ColBERT & 0.383 & 0.378 \\
        Separable LITE & \best{0.388} & \best{0.393} \\
        \bottomrule
    \end{tabular}}
\end{table}

For context, the {\tt T2} teacher~\citep{Hofstatter:2020} achieves a Dev MRR@10 of 0.399.
Even a CE student (with 6 layers and token dimension 768) cannot match this teacher performance: the best MRR@10 we get is 0.395.

We also note that separable LITE get good results for both the KL loss and margin MSE loss, while other scorers seem to prefer only one loss.
It is interesting to understand the effects of loss functions.

\begin{table}[h]
    \centering
    \caption{Natural Questions Dev MRR@10. Separable LITE achieves the best results both in direct training and distillation settings.}
    \label{tab:nq_mrr}
    \vskip 0.1in
    \scalebox{0.95}{
    \renewcommand{\arraystretch}{1.0}
    \begin{tabular}{lccc}
        \toprule
        \textbf{Scorer} & \textbf{Cross Entropy (one-hot labels)} & \textbf{KL (distillation)} & \textbf{Margin MSE} \\
        \midrule
        DE & 0.678 & 0.699 & 0.699 \\
        ColBERT & 0.690 & \best{0.754} & 0.756 \\
        Separable LITE & \best{0.710} & 0.741 & \best{0.769} \\
        \bottomrule
    \end{tabular}}
\end{table}

\Cref{tab:nq_mrr} includes results on NQ.
Here we report results in two settings: direct training with 1-hot labels and the cross entropy loss, and distillation training with the KL loss and margin MSE loss.
Separable LITE achieves the best results for both the cross-entropy loss and margin MSE loss; although ColBERT performs better with the KL loss, it gives lower scores than the margin MSE loss.

\subsection{Model ablations}
\label{sec:ablations}

\paragraph{Using top-$k$ aligned document tokens in ColBERT.}

Given query Transformer embedding vectors $\mathbf{q}_1,\ldots,\mathbf{q}_{L_1}$ and document Transformer embedding vectors $\mathbf{d}_1,\ldots,\mathbf{d}_{L_2}$, recall that ColBERT performs a sum-max reduction:
$$ \sum_{i \in [L_1]} \max_{j \in [L_2]} \mathbf{q}_i^\top \mathbf{d}_j. $$
In other words, for each query token $\mathbf{q}_i$, ColBERT finds the most-aligned document embedding vector and includes their dot-product in the score. \citet{qian2022multi} suggest using top-$k$ aligned document tokens for each query token; here we try $k=2,4,8$ on MS MARCO, but do not notice significant improvement compared with $k=1$.

\begin{table}[h]
    \centering
    \begin{tabular}{ccccc}
        \hline
        $k$ & 1 & 2 & 4 & 8 \\
        \hline
        MRR@10 & 0.383 & 0.378 & 0.380 & 0.382 \\
        \hline
    \end{tabular}
    \caption{Dev MRR@10 on MS MARCO with different values of $k$. We find that $k=1$ (i.e., the original ColBERT) is better than other options we try ($k=2,4,8$).}
    \label{tab:my_label}
\end{table}

\paragraph{Freezing query and document encoders.}

Recall that we use pretrained BERT models for query and document encoding, and moreover in all experiments above we also fine-tune the pretrained Transformers on MS MARCO and NQ. Here we explore performance of different scorers when the query and document Transformer encoders are frozen (i.e., pre-trained but not fine-tuned on MS MARCO).

When the query and document encoders are frozen, ColBERT does not require any additional fine-tuning since the sum-max function does not include any weights. In this case, ColBERT can achieve Dev MRR@10 score 0.112 on MS MARCO. 

For separable LITE, if we freeze the query and document Transformer encoders and only fine tune the separable LITE scorer (i.e., $\mathbf{W}_1,\mathbf{b}_1,\mathbf{W}_2,\mathbf{b}_2,\mathbf{W}_3,\mathbf{b}_3,\mathbf{W}_4,\mathbf{b}_4$ in \eqref{eq:sep_lite_row} and \eqref{eq:sep_lite_col}), then it can achieve Dev MRR@10 score 0.188 on MS MARCO, which is much better than ColBERT.

\subsection{KNRM results}
\label{sec:knrm_results}

For KNRM, following \cite{Xiong:2017}, we use $K=11$ kernels, where $\mu_1=0.9$, $\mu_2=0.7$, $\ldots$, $\mu_{10}=-0.9$ with $\sigma_1=\cdots=\sigma_{10}=0.1$, and $\mu_{11}=1.0$ with $\sigma_{11}=10^{-3}$.
We hold $\mu_k$ and $\sigma_k$ fixed and only train $\mathbf{w}$.

We report MRR@10 and nDCG@10 scores on in-domain tasks in \Cref{tab:knrm_in_domain}.
KNRM achieves similar scores to ColBERT overall, while separable LITE is more accurate than KNRM on all benchmarks.

\begin{table*}[t]
    \centering
    \caption{MRR@10 and nDCG@10 scores for in-domain tasks. KNRM is similar to ColBERT overall, while worse than separable LITE on all tasks.}
    \label{tab:knrm_in_domain}
    \begin{tabular}{lcccccccc}
        \toprule
         & \multicolumn{2}{c}{\textbf{MS MARCO}} & \multicolumn{2}{c}{\textbf{DL 2019}} & \multicolumn{2}{c}{\textbf{DL 2020}} & \multicolumn{2}{c}{\textbf{NQ}} \\
        \textbf{Scorer} & \textbf{MRR} & \textbf{nDCG} & \textbf{MRR} & \textbf{nDCG} & \textbf{MRR} & \textbf{nDCG} & \textbf{MRR} & \textbf{nDCG} \\
        \midrule
        ColBERT & 0.383 & 0.442 & 0.878 & 0.753 & 0.860 & 0.731 & 0.756 & 0.689 \\
        KNRM & 0.390 & 0.448 & 0.859 & 0.744 & 0.858 & 0.730 & 0.759 & 0.682 \\
        Sep LITE & \best{0.393} & \best{0.452} & \best{0.898} & \best{0.765} & \best{0.873} & \best{0.756} & \best{0.769} & \best{0.693} \\
        \bottomrule
    \end{tabular}    
\end{table*}

Moreover, separable LITE is much better than KNRM on zero-shot transfer: it is better than KNRM on 12 out of 14 datasets, as shown in \Cref{tab:beir_ndcg_knrm}.

\begin{table}[!t]
    \centering
    \caption{BEIR nDCG@10. Separable LITE is better than KNRM on 12 out of 14 datasets.}
    \label{tab:beir_ndcg_knrm}
    \renewcommand{\arraystretch}{1.0}
    \begin{tabular}{@{}lccccc}
        \toprule
        \textbf{Dataset} & \textbf{KNRM} & \textbf{Separable LITE} \\
        \midrule
        T-COVID & 0.741 & \best{0.763} \\
        NFCorpus & 0.353 & \best{0.358} \\
        NQ & 0.526 & \best{0.540} \\
        HotpotQA & 0.678 & \best{0.681} \\
        FiQA-2018 & 0.328 & \best{0.336} \\
        ArguAna & \best{0.446} & 0.424 \\
        Touch\'{e}-2020 & 0.301 & \best{0.305} \\
        CQAD & 0.367 & \best{0.374} \\
        Quora & 0.239 & \best{0.839} \\
        DBPedia & 0.420 & \best{0.434} \\
        SCIDOCS & 0.159 & \best{0.164} \\
        FEVER & 0.715 & \best{0.788} \\
        C-FEVER & 0.199 & \best{0.213} \\
        SciFact & \best{0.645} & 0.633 \\
        \bottomrule
    \end{tabular}    
\end{table}

\subsection{BEIR results of small separable LITE}

\Cref{tab:beir_ndcg_small_sep} shows BEIR results for small separable LITE introduce in \Cref{tab:latency}.
It is better than ColBERT on 8 out of 14 datasets.

\begin{table}[!t]
    \centering
    \caption{BEIR nDCG@10. Separable LITE is better than ColBERT on 11 out of 14 datasets.}
    \label{tab:beir_ndcg_small_sep}
    \begin{tabular}{@{}lcc|c}
        \toprule
        \textbf{Dataset} & \textbf{ColBERT} & \textbf{Small sep LITE} \\
        \midrule
        T-COVID & 0.761 & \best{0.767} \\
        NFCorpus & \best{0.356} & 0.353 \\
        NQ & 0.525 & \best{0.538} \\
        HotpotQA & \best{0.685} & 0.680 \\
        FiQA-2018 & \best{0.330} & 0.329 \\
        ArguAna & \best{0.433} & 0.433 \\
        Touch\'{e}-2020 & 0.274 & \best{0.298}  \\
        CQAD & 0.363 & \best{0.374} \\
        Quora & 0.767 & \best{0.836} \\
        DBPedia & 0.410 & \best{0.436} \\
        SCIDOCS & 0.155 & \best{0.163} \\
        FEVER & \best{0.782} & 0.772 \\
        C-FEVER & 0.190 & \best{0.214} \\
        SciFact & \best{0.667} & 0.622 \\
        \bottomrule
    \end{tabular}    
\end{table}

\section{Proof of \Cref{fact:lite_univ_approx}}
\label{sec:lite_univ_approx_proof}

Here we prove \Cref{fact:lite_univ_approx}.
We first restate it here and also include a universal approximation result without positional encodings.
\begin{theorem}[Universal approximation with LITE]\label{fact:lite_univ_approx_full}
    Let $s:\R^{(P\times L_1)\times(P\times L_2)}\to\R$ denote a continuous scoring function with a compact support $\Omega$ and $L_1,L_2\ge2$.
    For any $\mathcal{F}\in\{\mathcal{F}_{\rm f},\mathcal{F}_{\rm s}\}$ and any $\epsilon>0$, there exists a query Transformer $T_1:\R^{P\times L_1}\to\R^{P\times L_1}$, a document Transformer $T_2:\R^{P\times L_2}\to\R^{P\times L_2}$, and a scorer $f\in\mathcal{F}$, such that 
    \begin{align*}
        \int_\Omega\left(f\left(T_1(\mathbf{X})^\top T_2(\mathbf{Y})\right)-s(\mathbf{X},\mathbf{Y})\right)^2{\rm d}(\mathbf{X},\mathbf{Y})\le\epsilon.
    \end{align*}
    
    Under the same conditions, there also exist positional encoding matrices $\mathbf{E}\in\R^{P\times L_1}$ and $\mathbf{F}\in\R^{P\times L_2}$, a query Transformer $T_1:\R^{P\times L_1}\to\R^{P\times L_1}$ and a pooling function ${\sf pool}_1:\R^{P\times L_1}\to\R^{P\times2}$, a document Transformer $T_2:\R^{P\times L_2}\to\R^{P\times L_2}$ and a pooling function ${\sf pool}_2:\R^{P\times L_2}\to\R^{P\times2}$, and a scorer $f\in\mathcal{F}$, such that 
    \begin{align*}
        \int_\Omega\left(f\left({\sf pool}_1(T_1(\mathbf{X}+\mathbf{E}))^\top {\sf pool}_2(T_2(\mathbf{Y}+\mathbf{F}))\right)-s(\mathbf{X},\mathbf{Y})\right)^2{\rm d}(\mathbf{X},\mathbf{Y})\le\epsilon.
    \end{align*}
\end{theorem}
Our proof is based on the analysis of \citep{Yun:2020}: they showed that Transformer networks are universal approximators of continuous and compactly-supported sequence-to-sequence functions.
In our case, we need to show universal approximation with the dot-product matrix; to this end, we actually need a few technical lemmas from \citep{Yun:2020}, as detailed below.

Without loss of generality, we assume the support of the ground-truth scoring function is contained in $[0,1)^{P \times L_1} \times [0,1)^{P \times L_2}$.
The first step is to replace the ground-truth scoring function $s$ with a piece-wise constant function: let $\delta>0$ be small enough, and let 
\begin{align}\label{eq:s_delta}
    s_\delta(\mathbf{X},\mathbf{Y}):=\sum_{\mathbf{X}' \in \mathbb{G}_\delta, \mathbf{Y}' \in \mathbb{H}_\delta}s(\mathbf{X}',\mathbf{Y}')\mathds{1}\left[\mathbf{X}\in\mathbb{C}_{\mathbf{X}'} \textrm{ and }\mathbf{Y}\in\mathbb{C}_{\mathbf{Y}'}\right],
\end{align}
where $\mathbf{X}\in[0,1)^{P \times L_1}$, and $\mathbf{Y}\in[0,1)^{P \times L_2}$, and $\mathbb{G}_\delta:=\{0,\delta,\ldots,1-\delta\}^{P \times L_1}$, and $\mathbb{H}_\delta:=\{0,\delta,\ldots,1-\delta\}^{P \times L_2}$, and $\mathbb{C}_{\mathbf{X}'}:=\prod_{j=1}^P\prod_{k=1}^{L_1}[X'_{j,k},X'_{j,k}+\delta)$, and $\mathbb{C}_{\mathbf{Y}'}:=\prod_{j=1}^P\prod_{k=1}^{L_2}[Y'_{j,k},Y'_{j,k}+\delta)$.
Since $s$ is continuous, if $\delta$ is small enough, it holds that $s_\delta$ is a good approximation of $s$.

Next we follow \citep{Yun:2020} and try to approximate $s_\delta$ using LITE models based on \emph{modified} Transformers.
Recall that a standard Transformer uses softmax in attention layers and ReLU activation in MLPs; by contrast, in a modified Transformer, we use hardmax in attention layers, and in MLPs we are allowed to use activation functions from $\Phi$ which consists of piece-wise linear functions with at most three pieces where at least one piece is a constant.
Such a modified Transformer can then be approximated by a standard Transformer \citep[Lemma 9]{Yun:2020}.

Here are two key lemmas from \citep{Yun:2020}.
For simplicity, we state them for the query Transformer, but they will also be applied to the document Transformer.

The following lemma ensures that there exists a modified Transformer that can quantize the input domain, and thus we can just work with $\mathbb{G}_\delta$.
Similarly, on the document side, we can focus on $\mathbb{H}_\delta$.
\begin{lemma}[\citep{Yun:2020} Lemma 5]
\label{fact:quantization}
There exists a feedforward network $g_{\mathrm{q}}:[0,1)^{P \times L_1} \to \mathbb{G}_\delta$ with activations from $\Phi$, such that for any entry $1 \le i\le P$ and any $1\le j\le L_1$, it holds that $g_{\mathrm{q}}(\mathbf{X})_{i,j}=k\delta$ if $X_{i,j}\in[k\delta,(k+1)\delta)$, $k=0,\ldots,1/\delta-1$.
\end{lemma}

The following lemma ensures the existence of a modified Transformer that can implement a ``contextual mapping'': roughly speaking, it means each token of the Transformer output is a a unique Hash encoding of the whole input token sequence.
Below is a formal statement.
\begin{lemma}[\citep{Yun:2020} Lemma 6]
\label{fact:context}
Consider the following subset of $\mathbb{G}_\delta$:
\begin{align*}
    \widetilde{\mathbb{G}}_\delta:=\left\{\mathbf{X}\in\mathbb{G}_\delta \middle| \mathbf{X}_{:,i}\ne\mathbf{X}_{:,j}\textrm{ for all }i\ne j\right\}.
\end{align*}
If $L_1\ge2$ and $\delta\le1/2$, then there exists an attention network $g_{\mathrm{c}}:\mathbb{R}^{P \times L_1}\to \mathbb{R}^{P \times L_1}$ with the hardmax operator, a vector $\mathbf{u}\in\mathbb{R}^P$, constants $t_l,t_r$ with $0<t_l<t_r$, such that $\alpha(\mathbf{X}):=\mathbf{u}^\top g_{\mathrm{c}}(\mathbf{X})$ satisfies the following conditions:
\begin{enumerate}
    \item For any $\mathbf{X}\in\widetilde{\mathbb{G}}_\delta$, all entries of $\alpha(\mathbf{X})$ are different.
    
    \item For any $\mathbf{X},\mathbf{X}'\in\widetilde{\mathbb{G}}_\delta$ such that $\mathbf{X}'$ is not a permutation of $\mathbf{X}$, all entries of $\alpha(\mathbf{X})$, $\alpha(\mathbf{X}')$ are different.
    
    \item For any $\mathbf{X}\in\widetilde{\mathbb{G}}_\delta$, all entries of $\alpha(\mathbf{X})$ are in $[t_l,t_r]$.
    
    \item For any $\mathbf{X}\in\mathbb{G}_\delta\setminus\widetilde{\mathbb{G}}_\delta$, all entries of $\alpha(\mathbf{X})$ are outside $[t_l,t_r]$.
\end{enumerate}
\end{lemma}

For the document side, consider 
\begin{align*}
    \widetilde{\mathbb{H}}_\delta:=\left\{\mathbf{X}\in\mathbb{H}_\delta \middle| \mathbf{Y}_{:,i}\ne\mathbf{Y}_{:,j}\textrm{ for all }i\ne j\right\}.
\end{align*}
\Cref{fact:context} also ensures the existence of an attention network $h_{\mathrm{c}}:\mathbb{R}^{P \times L_2}\to \mathbb{R}^{P \times L_2}$ with the hardmax operator, a vector $\mathbf{v}\in\mathbb{R}^P$, constants $s_l,s_r$ with $0<s_l<s_r$, such that $\beta(\mathbf{Y}):=\mathbf{v}^\top h_{\mathrm{c}}(\mathbf{Y})$ satisfies similar conditions.
Also note that for small enough $\delta$, we can neglect $\mathbb{G}_\delta\setminus\widetilde{\mathbb{G}}_\delta$ and $\mathbb{H}_\delta\setminus\widetilde{\mathbb{H}}_\delta$, since 
$|\mathbb{G}_\delta\setminus\widetilde{\mathbb{G}}_\delta|=O\left(\delta^P|\mathbb{G}_\delta|\right)$ and $|\mathbb{H}_\delta\setminus\widetilde{\mathbb{H}}_\delta|=O\left(\delta^P|\mathbb{H}_\delta|\right)$.

Now we are ready to prove \Cref{fact:lite_univ_approx_full}.
We first consider the case without positional encodings.

\paragraph{Analysis without positional encodings.}

Note that for $\mathbf{X}\in\widetilde{\mathbb{G}}_\delta$ and $\mathbf{Y}\in\widetilde{\mathbb{H}}_\delta$, it holds that $\alpha(\mathbf{X})$ and $\beta(\mathbf{Y})$ already include enough information to determine the score. However, in LITE models, the final score is calculated only based on dot products between query embedding vectors and document embedding vectors. 
As a result, we need to first insert $\mathbf{u}$ and $\mathbf{v}$ into the Transformer embeddings.
The following lemma handles this issue: there exists a feedforward network such that for each $\mathbf{X}\in\widetilde{\mathbb{G}}_\delta$, it replaces one token in $g_{\mathrm{c}}(\mathbf{X})$ with $\mathbf{v}$ while keeping other tokens unchanged.

\begin{lemma}
\label{fact:embed_uv}
Consider the activation function $\varphi$ with $\varphi(z)=1$ if $0\le z\le1$, and $\varphi(z)=0$ if $z<0$ or $z>1$.
There exists a feedforward network $g_{\mathrm{v}}:\mathbb{R}^P \to \mathbb{R}^P$ with activation $\varphi$ such that for any $\mathbf{X}\in\widetilde{\mathbb{G}}_\delta$, let $i:=\argmin_j\alpha(\mathbf{X})_j$, then $g_{\mathrm{v}}(g_{\mathrm{c}}(\mathbf{X})_{:,i})=\mathbf{v}$, while for $j\ne i$, it holds that $g_{\mathrm{v}}(g_{\mathrm{c}}(\mathbf{X})_{:,j})=g_{\mathrm{c}}(\mathbf{X})_{:,j}$.
\end{lemma}
\begin{proof}
For any $\mathbf{X}\in\widetilde{\mathbb{G}}_\delta$ and any $i$, $1\le i\le L_1$, \Cref{fact:context} ensures that there exists constants $l(\mathbf{X},i)$ and $r(\mathbf{X},i)$ such that $0<l(\mathbf{X},i)<\alpha(\mathbf{X})_i<r(\mathbf{X},i)$, and that $[l(\mathbf{X},i),r(\mathbf{X},i)]$ does not contain other entries in $\alpha(\mathbf{X})$, and moreover $[l(\mathbf{X},i),r(\mathbf{X},i)]$ does not contain entries from $\alpha(\mathbf{X}')$ for $\mathbf{X}'\in\widetilde{\mathbb{G}}_\delta$ which is not a permutation of $\mathbf{X}$.
For this $(\mathbf{X},i)$ pair, if $i:=\argmin_j\alpha(\mathbf{X})_j$, we construct the following neuron
\begin{align*}
    \psi_{\mathbf{X},i}(\mathbf{z}):=\varphi\left(\frac{1}{r(\mathbf{X},i)-l(\mathbf{X},i)}\left(\mathbf{u}^\top\mathbf{z}-l(\mathbf{X},i)\right)\right)\mathbf{v},
\end{align*}
otherwise let
\begin{align*}
    \psi_{\mathbf{X},i}(\mathbf{z}):=\varphi\left(\frac{1}{r(\mathbf{X},i)-l(\mathbf{X},i)}\left(\mathbf{u}^\top\mathbf{z}-l(\mathbf{X},i)\right)\right)g_{\mathrm{c}}(\mathbf{X})_{:,i}.
\end{align*}
The full network is the sum of all such neurons
\begin{align*}
    g_{\mathrm{v}}(\mathbf{z}):=\sum_{\mathbf{X}\in\widetilde{\mathbb{G}}_\delta,1\le i\le L_1}\psi_{\mathbf{X},i}(\mathbf{z}),
\end{align*}
which satisfies the requirement of \Cref{fact:embed_uv}.
\end{proof}

\Cref{fact:embed_uv} is stated for the query side; on the document side, it also follows that there exists a feedforward network $h_{\mathrm{u}}$ that can replace one token in the embeddings given by $h_{\mathrm{c}}$ by $\mathbf{u}$.
Then we are ready to prove \Cref{fact:lite_univ_approx_full} without positional encodings.
\begin{proof}[Proof of \Cref{fact:lite_univ_approx_full}, no positional encodings]
In this proof, we will focus on $\mathbf{X}\in\widetilde{\mathbb{G}}_\delta$ and $\mathbf{Y}\in\widetilde{\mathbb{H}}_\delta$ as ensured by \Cref{fact:quantization,fact:context}.
We also use notation introduced in \Cref{fact:context,fact:embed_uv}.

First consider $\mathbf{u}$ and $\mathbf{v}$ given by \Cref{fact:context}.
Without loss of generality, we can assume $\mathbf{u}^\top\mathbf{v}\le0$; if $\mathbf{u}^\top\mathbf{v}>0$, we will replace $\mathbf{v}$ with $-\mathbf{v}$ and replace $h_{\rm c}(\mathbf{Y})$ with $-h_{\rm c}(\mathbf{Y})$, which ensures $\mathbf{u}^\top\mathbf{v}\le0$, and moreover the conclusions of \Cref{fact:context} still hold.
In detail, in the construction of $g_{\rm v}$, we use $-\mathbf{v}$ instead of $\mathbf{v}$, while in the construction of $h_{\rm u}$, we use $-h_{\rm c}(\mathbf{Y})$ instead of $h_{\rm c}(\mathbf{Y})$.
As a result, in the following we assume $\mathbf{u}^\top\mathbf{v}\le0$.

Recall that for $\mathbf{X}\in\widetilde{\mathbb{G}}_\delta$, the range of $\mathbf{u}^\top g_{\rm c}(\mathbf{X})$ is denoted by $[t_l,t_r]$ with $0<t_l<t_r$, while for $\mathbf{Y}\in\widetilde{\mathbb{H}}_\delta$, the range of $\mathbf{v}^\top h_{\rm c}(\mathbf{Y})$ is denoted by $[s_l,s_r]$ with $0<s_l<s_r$.
Define
\begin{align*}
    M:=\max_{\mathbf{X}\in\widetilde{\mathbb{G}}_\delta}\max_{\mathbf{Y}\in\widetilde{\mathbb{H}}_\delta}\max_{i,j}\left|g_{\rm c}(\mathbf{X})_{:,i}^\top h_{\rm c}(\mathbf{Y})_{:,j}\right|.
\end{align*}
In the following, we will assume $t_l>M$ and $s_l>t_r$ without loss of generality; if these conditions do not hold, we can let $\lambda_1,\lambda_2$ be large enough such that $\lambda_1t_l>M$ and $\lambda_2s_l>\lambda_1t_r$, and scale $\mathbf{u}$ to $\lambda_1\mathbf{u}$, and scale $\mathbf{v}$ to $\lambda_2\mathbf{v}$.

Given $\mathbf{X}\in\widetilde{\mathbb{G}}_\delta$ and $\mathbf{Y}\in\widetilde{\mathbb{H}}_\delta$, we consider $\mathbf{Q}=g_{\rm v}(g_{\rm c}(\mathbf{X}))\in\R^{P \times L_1}$, and $\mathbf{D}=h_{\rm u}(h_{\rm c}(\mathbf{Y}))\in\R^{P \times L_2}$, and the dot-product matrix $\mathbf{S}:=\mathbf{Q}^\top\mathbf{D}\in\R^{L_1\times L_2}$.
\Cref{fact:embed_uv} ensures that $\mathbf{Q}$ has one column equal to $\mathbf{v}$, while $\mathbf{D}$ has one column equal to $\mathbf{u}$.

Let $\mathbf{q}$ denote an arbitrary column of $\mathbf{Q}$ other than $\mathbf{v}$, and let $\mathbf{d}$ denote an arbitrary column of $\mathbf{D}$ other than $\mathbf{u}$.
Due to previous discussion, we have $\mathbf{v}^\top\mathbf{d}\ge s_l>t_r\ge\mathbf{q}^\top\mathbf{u}$, and therefore we can distinguish them.
Additionally $\mathbf{q}^\top\mathbf{u}\ge t_l>M$, and thus we can distinguish it from other entries of $\mathbf{S}$, including $\mathbf{v}^\top\mathbf{u}\le0$.

Now let us examine $\mathbf{S}$ in detail.
Suppose $\mathbf{Q}_{:,i}=\mathbf{v}$ and $\mathbf{D}_{:,j}=\mathbf{u}$ for some $1\le i\le L_1$ and $1\le j\le L_2$.
Then 
\begin{align*}
    \mathbf{S}_{i,:}=(\mathbf{Q}^\top\mathbf{D})_{i,:}=[\mathbf{v}^\top\mathbf{d}_1,\cdots,\mathbf{v}^\top\mathbf{u},\cdots,\mathbf{v}^\top\mathbf{d}_{L_2}],
\end{align*}
and
\begin{align*}
    \mathbf{S}_{:,j}=[\mathbf{q}_1^\top\mathbf{u},\cdots,\mathbf{v}^\top\mathbf{u},\cdots,\mathbf{q}_{L_1}^\top\mathbf{u}]^\top.
\end{align*}
The previous scaling allows us to find $\mathbf{S}_{i,:}$ and $\mathbf{S}_{:,j}$.
\Cref{fact:context} ensures that every element of $\mathbf{S}_{i,:}$ other than $\mathbf{v}^\top\mathbf{u}$ can uniquely determine the set of columns of the document input $\mathbf{Y}$, but not the order of columns since Transformers without positional encodings are permutation-equivariant \citep[Claim 1]{Yun:2020}.
However, all elements of $\mathbf{S}_{i,:}$ together are able to determine the exact order of columns of $\mathbf{Y}$.
Similarly, $\mathbf{S}_{:,j}$ as a whole can determine the exact query input $\mathbf{X}$, including the order of columns.
Consequently, $\mathbf{S}$ can uniquely determine the input pair $(\mathbf{X},\mathbf{Y})$, and also the ground-truth score $s_\delta(\mathbf{X},\mathbf{Y})$.

For flattened LITE, note that $\widetilde{\mathbb{G}}_\delta$ and $\widetilde{\mathbb{H}}_\delta$ are both finite, and thus the set of possible dot-product matrix 
\begin{equation*}
    \left\{\mathbf{Q}^\top\mathbf{D}\middle|\mathbf{Q}=g_{\rm v}(g_{\rm c}(\mathbf{X})),\mathbf{D}=h_{\rm u}(h_{\rm c}(\mathbf{Y})),\mathbf{X}\in\widetilde{\mathbb{G}}_\delta,\mathbf{Y}\in\widetilde{\mathbb{H}}_\delta\right\}
\end{equation*}
is also finite. 
Moreover, each dot-product matrix uniquely determines the ground-truth score, as discussed above.
Therefore there exists a 2-layer ReLU network that uniformly approximates an interpolations of these scores \citep{cybenko1989approximation,funahashi1989approximate,hornik1989multilayer}, which finishes the proof.
    
For separable LITE, recall that we first apply a shared MLP $f_1$ to reduce every row of $\mathbf{S}$ to a scalar, and thus get a column vector; then we apply another MLP $f_2$ to reduce this column vector to a final score.
Now let $\psi$ denote an injection from $\widetilde{\mathbb{H}}_\delta$ to $[t_r+1,t_r+2]$, i.e., for any $\mathbf{Y},\mathbf{Y}'\in\widetilde{\mathbb{H}}_\delta$, we have $\psi(\mathbf{Y}),\psi(\mathbf{Y}')\in[t_r+1,t_r+2]$, and $\psi(\mathbf{Y})\ne\psi(\mathbf{Y}')$.
There exists such a $\psi$ since $\widetilde{\mathbb{H}}_\delta$ is finite.

Now if the $i$-th column of $\mathbf{Q}$ is $\mathbf{v}$, then we let $f_1$ map $\mathbf{S}_{i,:}$ to $\psi(\mathbf{Y})$; this is well-defined since $\mathbf{S}_{i,:}$ uniquely determines $\mathbf{Y}$, as discussed above.
For any $i'\ne i$, we let $f_1$ map $\mathbf{S}_{i',:}$ to $\mathbf{q}_{i'}^\top\mathbf{u}\in[t_l,t_r]$.
Note that by our construction, $f_1(\mathbf{S}_{i,:})\ge t_r+1>t_r\ge f_1(\mathbf{S}_{i',:})$.
As a result, $f_1(\mathbf{S})$ can uniquely determines $(\mathbf{X},\mathbf{Y})$, and thus there exists another MLP $f_2$ which can approximate the ground-truth score $s_\delta$.
\end{proof}

\paragraph{Analysis with positional encodings.}

Here we consider the case with positional encodings.
Following \citep{Yun:2020}, we will use fixed positional encodings: let $\mathbf{1}$ denote the $P$-dimensional all-ones vector, and let $\mathbf{E}\in\R^{P\times L_1}$ denote the matrix whose $j$-th column is given by $(j-1)\mathbf{1}$, and similarly let $\mathbf{F}\in\R^{P\times L_2}$ denote the matrix whose $j$-th column is given by $(j-1)\mathbf{1}$.
Given input $\mathbf{X}\in[0,1)^{P\times L_1}$ and $\mathbf{Y}\in[0,1)^{P\times L_2}$, we transform them to $(\mathbf{X}+\mathbf{E})/L_1$ and $(\mathbf{Y}+\mathbf{F})/L_2$.
Note that after the transformation, it holds that $(\mathbf{X}+\mathbf{E})/L_1\in\prod_{i=1}^P\prod_{j=1}^{L_1}[(j-1)/L_1,j/L_1)$; in other words, different columns of $(\mathbf{X}+\mathbf{E})/L_1$ have different ranges.

We can now invoke our earlier analysis.
Let $\delta=1/(nL_1L_2)$ for some large enough integer $n$ such that the approximation error in \eqref{eq:s_delta} is small enough.
Then \Cref{fact:quantization} implies there exist feedforward networks $g_{\mathrm{q}}$ and $h_{\mathrm{q}}$ that can quantize the input domains to $\mathbb{G}_\delta=\{0,\delta,\cdots,1-\delta\}^{P\times L_1}$ and  $\mathbb{H}_\delta=\{0,\delta,\cdots,1-\delta\}^{P\times L_2}$.
Combined with the positional encodings, we only need to consider the following input domains:
\begin{align*}
    \mathbb{G}_{\delta,\mathrm{pe}} & :=\left\{g_{\mathrm{q}}((\mathbf{X}+\mathbf{E})/L_1)\middle|\mathbf{X}\in[0,1)^{P\times L_1}\right\},\\
    \mathbb{H}_{\delta,\mathrm{pe}} & :=\left\{h_{\mathrm{q}}((\mathbf{Y}+\mathbf{F})/L_2)\middle|\mathbf{Y}\in[0,1)^{P\times L_2}\right\}.
\end{align*}
Note that for any $\mathbf{X}\in\mathbb{G}_{\delta,\mathrm{pe}}$, all of its columns are different, and for any different $\mathbf{X},\mathbf{X}'\in\mathbb{G}_{\delta,\mathrm{pe}}$, it holds that the columns of $\mathbf{X}$ are not a permutation of the columns of $\mathbf{X}'$.

Then we can invoke \Cref{fact:context}, which shows the existence of an attention network $g_{\mathrm{c}}$ and a vector $\mathbf{u}$ such that for any $\mathbf{X}\in\mathbb{G}_{\delta,\mathrm{pe}}$, it holds that any entry of $\mathbf{u}^\top g_{\mathrm{c}}(\mathbf{X})$ uniquely determines $\mathbf{X}$.
Similarly, there exists $h_{\mathrm{c}}$ and $\mathbf{v}$ which implement contextual mapping for documents.
Now we just need the following pooling functions: for the query side, the pooling function outputs $\mathbf{v}$ and $g_{\mathrm{c}}(\mathbf{X})_{:,1}$; for the document side, the pooling function outputs $\mathbf{u}$ and $h_{\mathrm{c}}(\mathbf{Y})_{:,1}$.
The similarity matrix is then given by
\begin{equation*}
    \begin{bmatrix}
        \mathbf{u}^\top\mathbf{v} & \mathbf{v}^\top h_{\mathrm{c}}(\mathbf{Y})_{:,1} \\
        \mathbf{u}^\top g_{\mathrm{c}}(\mathbf{X})_{:,1} & g_{\mathrm{c}}(\mathbf{X})_{:,1}^\top h_{\mathrm{c}}(\mathbf{Y})_{:,1}
    \end{bmatrix}
\end{equation*}
In particular, the off-diagonal entries of the similarity matrix are enough to determine the query-document pair.
Therefore we can further use MLP scorers to approximate the ground-truth scoring function.

\section{Proof of \Cref{fact:de_neg}}
\label{sec:de_neg}

To prove \Cref{fact:de_neg}, we first construct an empirical dataset on which we show a simple dot-product dual encoder has a large approximation error based on a rank argument.
This empirical dataset can then be extended to a distribution on $[0,1]^{P\times L}$.

Here we let $L_1=L_2=L$, i.e., all queries and documents have the same number of tokens.
The set of queries is simply $\mathcal{Q}:=\{0,1\}^{P\times L}$, i.e., there are $2^{PL}$ queries, each of them has dimension $P\times L$, and each coordinate of them can be either $0$ or $1$.
The set of documents is also given by $\mathcal{D}:=\{0,1\}^{P\times L}$.
Given a query $\mathbf{X}\in\mathcal{Q}$ and a document $\mathbf{Y}\in\mathcal{D}$, define the ground-truth score as
\begin{equation}\label{eq:K^*_emp}
    K^*(\mathbf{X},\mathbf{Y}):=\mathrm{tr}(\mathbf{X}^\top\mathbf{Y})
\end{equation}
Let $\mathbf{K}^*\in\R^{2^{PL} \times 2^{PL}}$ denote the matrix of ground-truth scores between all query-document pairs.
We will show the following result.
\begin{lemma}
\label{fact:de_neg_emp}
Let $T_1:\R^{P\times L}\to\R^O$ denote an arbitrary function that maps a query $\mathbf{X}\in\mathcal{Q}$ to an $O$-dimensional vector, and let $T_2:\R^{P\times L}\to\R^O$ denote an arbitrary function that maps a document $\mathbf{Y}\in\mathcal{D}$ to an $O$-dimensional vector.
Given $\mathbf{X}\in\mathcal{Q}$ and $\mathbf{Y}\in\mathcal{D}$, define the dot-product DE score as $K^{\rm de}(\mathbf{X},\mathbf{Y})=T_1(\mathbf{X})^\top T_2(\mathbf{Y})$, and let $\mathbf{K}^{\rm de}\in\R^{2^{PL}\times 2^{PL}}$ denote the matrix of DE scores for all query-document pairs.
If $O\le PL-1$, then the mean square error between $\mathbf{K}^*$ and $\mathbf{K}^{\rm de}$ is at least $1/16$:
\begin{equation*}
   \frac{1}{2^{2PL}} \|\mathbf{K}^*-\mathbf{K}^{\rm de}\|_F^2\ge\frac{1}{16}.
\end{equation*}
\end{lemma}

To prove \Cref{fact:de_neg_emp}, we first show the following linear algebra fact.
\begin{proposition}
\label{fact:eig}
    Let $I_n$ denote the $n$-by-$n$ diagonal matrix, and let $J_n$ denote the $n$-by-$n$ matrix whose entries are all $1$. 
    For $\lambda>0$, the matrix $\lambda I_n+J_n$ has rank $n$; its top eigenvalue is $\lambda+n$, while the remaining $n-1$ eigenvalues are $\lambda$.
\end{proposition}
\begin{proof}
First consider the matrix $J_n$.
Let $\mathbf{1}_n$ denote the $n$-dimensional vector whose entries are all $1$; it is an eigenvector of $J_n$ with eigenvalue $n$.
Moreover, $J_n$ also has eigenvalue $0$; the corresponding eigenspace is given by $\{\mathbf{z}\in\R^n|\sum_iz_i=0\}$, which has dimension $n-1$.
As a result, the eigenvalue $0$ has multiplicity $n-1$.

Moreover, note that for any $n$-by-$n$ matrix $\mathbf{A}$ with eigenvalue $\mu$, the matrix $\lambda I_n+\mathbf{A}$ has an eigenvalue $\lambda+\mu$.
Consequently, the matrix $\lambda I_n+J_n$ has eigenvalue $\lambda+n$ with multiplicity $1$, and eigenvalue $\lambda$ with multiplicity $n-1$.
\end{proof}

Next we prove the following properties of $\mathbf{K}^*$ using \Cref{fact:eig}.
\begin{lemma}
\label{fact:K^*}
    It holds that $\mathbf{K}^*$ has rank $PL$; its top eigenvalue is $2^{PL-2}(PL+1)$, while the remaining $PL-1$ eigenvalues are $2^{PL-2}$.
\end{lemma}
\begin{proof}
Let $\mathbf{U}\in\R^{2^{PL}\times PL}$ denote the matrix whose rows are obtained by flattening elements of $\{0,1\}^{P\times L}$ (i.e., the query set $\mathcal{Q}$ and document set $\mathcal{D}$).
It then holds that $\mathbf{K}^*=\mathbf{U}\mathbf{U}^\top$.
We will analyze the spectrum of $\mathbf{K}^*$ by considering $\mathbf{U}^\top\mathbf{U}$, since it has the same eigenvalues as $\mathbf{U}\mathbf{U}^\top$.

We claim that $\mathbf{U}^\top\mathbf{U}=2^{PL-2}(I_{PL}+J_{PL})$.
First consider diagonal entries of $\mathbf{U}^\top\mathbf{U}$.
For any $1\le i\le PL$, it holds that $\mathbf{U}_{:,i}$ has half entries equal to $0$, and the other half entries equal to $1$.
As a result, $(\mathbf{U}^\top\mathbf{U})_{i,i}=2^{PL-1}$.
Next we consider off-diagonal entries of $\mathbf{U}^\top\mathbf{U}$.
For any $1\le i,j\le PL$ and $i\ne j$, it holds that $U_{k,i}=U_{k,j}=1$ for $1/4$ of all positions $k$; therefore $(\mathbf{U}^\top\mathbf{U})_{i,j}=2^{PL-2}$.
This proves our claim.

The claim of \Cref{fact:K^*} then follows from \Cref{fact:eig}.
\end{proof}

Now we can prove \Cref{fact:de_neg_emp}
\begin{proof}[Proof of \Cref{fact:de_neg_emp}]
Let $T_1:\R^{P\times L}\to\R^O$ denote an arbitrary mapping; in particular, it could represent a Transformer with positional encodings which maps a query $\mathbf{X}\in\mathcal{Q}$ to an $O$-dimensional embedding vector.
Furthermore, let $T_1(\mathcal{Q})\in\R^{2^{PL}\times O}$ denote the embeddings of all queries given by $T_1$.
Similarly, let $T_2:\R^{P\times L}\to\R^O$ denote an arbitrary mapping which represents the document encoder, and let $T_2(\mathcal{D})\in\R^{2^{PL}\times O}$ denote embeddings of all documents given by $T_2$.
The matrix of dot-product DE scores is then given by $\mathbf{K}^{\rm de}:=T_1(\mathcal{Q})T_2(\mathcal{D})^\top$.

By definition, $\mathbf{K}^{\rm de}$ has rank at most $O$.
If $O\le PL-1$, then \Cref{fact:K^*} implies that 
\begin{equation*}
   \frac{1}{2^{2PL}} \|\mathbf{K}^*-\mathbf{K}^{\rm de}\|_F^2\ge \frac{1}{2^{2PL}}(2^{PL-2})^2\ge\frac{1}{16}.
\end{equation*}
\end{proof}

Then we extend \Cref{fact:de_neg_emp} to \Cref{fact:de_neg}.
\begin{proof}[Proof of \Cref{fact:de_neg}]
Recall that the domain of the ground-truth score $K^*$ defined in \eqref{eq:K^*_emp} is $\{0,1\}^{P\times L}\times\{0,1\}^{P\times L}$.
We first extend its domain to $[0,1]^{P\times L}\times[0,1]^{P\times L}$ by quantizing the inputs: given $\mathbf{X}\in[0,1]^{P\times L}$, its quantized version $\widehat{\mathbf{X}}\in\{0,1\}^{P\times L}$ is obtained by mapping all entries less than $1/2$ to $0$ and other entries to $1$. 
Similarly, given $\mathbf{Y}\in[0,1]^{P\times L}$, we can define its quantized version $\widehat{\mathbf{Y}}\in\{0,1\}^{P\times L}$.
We then let $K^*(\mathbf{X},\mathbf{Y})=K^*(\widehat{\mathbf{X}},\widehat{\mathbf{Y}})=\mathrm{tr}(\widehat{\mathbf{X}}^\top\widehat{\mathbf{Y}})$.
Note that $K^*$ defined in this way is not yet continuous; later we will replace it with a continuous ground-truth function, but we will first use $K^*$ below since it simplifies the analysis.

Let $T_1:\R^{P\times L}\to\R^O$ and $T_2:\R^{P\times L}\to\R^O$ denote arbitrary mappings.
Let
\begin{equation*}
    \mathbb{M}:=\left\{\mathbf{Z}\in\R^{P\times L}\middle|Z_{i,j}=0\textrm{ or }1/2,1\le i\le P,1\le j\le L\right\}.
\end{equation*}
For $\mathbf{Z}\in\mathbb{M}$, let $\mathbb{C}_{\mathbf{Z}}:=\prod_{i=1}^{P}\prod_{j=1}^L[Z_{i,j},Z_{i,j}+1/2]$.

Now we want to find a lower bound on 
\begin{align}
     & \ \int_{\mathbf{X}\in[0,1]^{P\times L},\mathbf{Y}\in[0,1]^{P\times L}}\left(T_1(\mathbf{X})^\top T_2(\mathbf{Y})-K^*(\mathbf{X},\mathbf{Y})\right)^2{\rm d}\mathbf{X}{\rm d}\mathbf{Y} \nonumber \\
    = & \ \sum_{\mathbf{Z},\mathbf{Z}'\in\mathbb{M}}\int_{\mathbf{X}\in\mathbb{C}_{\mathbf{Z}},\mathbf{Y}\in\mathbb{C}_{\mathbf{Z}'}}\left(T_1(\mathbf{X})^\top T_2(\mathbf{Y})-K^*(\mathbf{X},\mathbf{Y})\right)^2{\rm d}\mathbf{X}{\rm d}\mathbf{Y} \nonumber \\
    = & \ \int_{\mathbf{X}\in\mathbb{C}_{\mathbf{0}},\mathbf{Y}\in\mathbb{C}_{\mathbf{0}}}\sum_{\mathbf{Z},\mathbf{Z}'\in\mathbb{M}}\left(T_1(\mathbf{X}+\mathbf{Z})^\top T_2(\mathbf{Y}+\mathbf{Z}')-K^*(\mathbf{X}+\mathbf{Z},\mathbf{Y}+\mathbf{Z}')\right)^2{\rm d}\mathbf{X}{\rm d}\mathbf{Y}, \label{eq:de_neg_tmp}
\end{align}
where we let $\mathbf{0}$ denotes the $P$-by-$L$ matrix whose entries are all $0$.
Note that in \eqref{eq:de_neg_tmp}, for any $\mathbf{X},\mathbf{Y}\in\mathbb{C}_{\mathbf{0}}$, the error can be lower bounded by $2^{2PL}/16$ using the proof of \Cref{fact:de_neg_emp}.
Therefore we have 
\begin{align*}
    \eqref{eq:de_neg_tmp} & \ge\int_{\mathbf{X}\in\mathbb{C}_{\mathbf{0}},\mathbf{Y}\in\mathbb{C}_{\mathbf{0}}}\frac{2^{2PL}}{16}{\rm d}\mathbf{X}{\rm d}\mathbf{Y} \\
     &=\frac{2^{2PL}}{16}\cdot{\sf vol}(\mathbb{C}_{\mathbf{0}})^2 \\
     & =\frac{1}{16}.
\end{align*}

As mentioned above, $K^*$ is not continuous, and the final step of the proof is to replace it with a continuous ground-truth function.
Previously, we quantize the input by transforming entries less than $1/2$ to $0$ and other entries to $1$. 
Now we use the following transformation function $\phi_\tau$: $\phi_\tau(z)=0$ if $z\le\frac{1}{2}-\tau$, and $\phi_\tau(z)=1$ if $z\ge\frac{1}{2}+\tau$, and otherwise $\phi_\tau(z)=\frac{1}{2}+\frac{1}{2\tau}(z-\frac{1}{2})$.
Given $\mathbf{X},\mathbf{Y}\in[0,1]^{P\times L}$, we apply $\phi_\tau$ to every entry of $\mathbf{X},\mathbf{Y}$ and get $\phi_\tau(\mathbf{X})$ and $\phi_\tau(\mathbf{Y})$, and define $K_\tau^*(\mathbf{X},\mathbf{Y})$
\begin{equation*}
    K_\tau^*(\mathbf{X},\mathbf{Y}):=\mathrm{tr}(\phi_\tau(\mathbf{X})^\top\phi_\tau(\mathbf{Y})).
\end{equation*}
Note that $K^*_\tau$ is continuous for any $\tau$, and as $\tau$ goes to $0$, it holds that $K^*_\tau$ becomes arbitrarily close to $K^*$ in $\ell_2$ distance. Therefore there exists a small enough $\tau$ such that $K^*_\tau$ satisfies the requirements of \Cref{fact:de_neg}.
\end{proof}

\end{document}